\documentclass[sigconf]{acmart}
\renewcommand\footnotetextcopyrightpermission[1]{} % removes footnote with conference information in first column

\makeatletter                   %1
\def\mdseries@tt{m}             %1
\makeatother                    %1

\usepackage[plain]{fancyref}
\usepackage[draft=true]{minted} %2
\usepackage{color}
\usepackage{hyperref}           %5
\hypersetup{
	colorlinks=true,
	linkcolor=blue,
	filecolor=red,      
	urlcolor=magenta,
	breaklinks=true,            %3
}
\usepackage{breakurl}           %3
\usepackage{booktabs}
\usepackage{enumitem}
\usepackage{float}
\usepackage{makecell}
\usepackage{multirow}
\usepackage{multicol}
\usepackage{pgfplots}
\usepackage{pgfplotstable}
\usepackage{pifont}
\usepackage{ragged2e}
\usepackage{subcaption}
\usepackage{supertabular}
\usepackage{tabu}
\usepackage{tabularx}
\usepackage{threeparttable}
\usepackage{tikz}
\usepackage{xcolor}
\usepackage{xspace}
\usepackage[framemethod=tikz]{mdframed} 
\usepackage{algpseudocode}
\newtheorem{theorem}{Theorem}[section]

\newtheorem{lemma}[theorem]{Lemma}

\theoremstyle{definition}

\newcommand{\xor}{\oplus}
\newcommand{\Xor}{\bigoplus}
\newcommand{\band}{\odot}

\newcommand{\cmark}{\ding{51}}
\newcommand{\xmark}{\ding{55}}

\newcommand{\iseq}{\ensuremath{\stackrel{?}{=} }}

\newcommand{\Z}[1]{\ensuremath{\mathbb{Z}}_{2^{#1}}}

\setlist[description]{style=unboxed,leftmargin=0cm}

\newenvironment{myitemize}{
	\begin{list}{{$\bullet$}}{
			\setlength\partopsep{0pt}
			\setlength\parskip{0pt}
			\setlength\parsep{0pt}
			\setlength\topsep{2pt}
			\setlength\itemsep{1pt}
			\setlength{\itemindent}{0pt}
			\setlength{\leftmargin}{9pt}
		}
	}{
		\vspace{1mm}
	\end{list}
}

\newcounter{itemcount}

\newcommand{\tabref}[1]{Table~\protect\ref{tab:#1}}
\newcommand{\secref}[1]{Section~\protect\ref{sec:#1}}

\newcommand{\figref}[1]{Figure~\ref{fig:#1}}
\newcommand{\figlab}[1]{\label{fig:#1}}

\newenvironment{boxfig*}[2]{% {#1}{#2} = {Caption}{label}
	\begin{figure*}[h!]		
		\fontsize{5}{5}\selectfont
		\newcommand{\FigCaption}{#1}
		\newcommand{\FigLabel}{#2}
		\vspace{-.05cm}
		\begin{center}
			\begin{small}			 
				\begin{adjustbox}{max width=\textwidth}
					\begin{tabular}{@{}|@{~~}l@{~~}|@{}}
						\hline
						\rule[-1ex]{0pt}{1ex}\begin{minipage}[b]{.95\linewidth}
							\vspace{1ex}	
						}{%
						\end{minipage}\\
						\hline
					\end{tabular}	
				\end{adjustbox}		
			\end{small}
			\vspace{-0.1cm}
			\caption{\FigCaption}
			\figlab{\FigLabel}
		\end{center}
		\vspace{-.38cm}
	\end{figure*}
}

\newenvironment{myboxfig*}[2]{% {#1}{#2} = {Caption}{label}
	\begin{figure*}[!htb]		
		\fontsize{5}{5}\selectfont
		\newcommand{\FigCaption}{#1}
		\newcommand{\FigLabel}{#2}
		\vspace{-.10cm}
		\begin{center}
			\caption{\FigCaption}
			\begin{small}			 
				\begin{adjustbox}{max width=\textwidth}
					\begin{tabular}{@{}|@{~~}l@{~~}|@{}}
						\hline
						\rule[-1ex]{0pt}{1ex}\begin{minipage}[b]{.95\linewidth}
							\vspace{1ex}	
						}{%
						\end{minipage}\\
						\hline
					\end{tabular}	
				\end{adjustbox}		
			\end{small}
			\vspace{-0.25cm}
			\figlab{\FigLabel}
		\end{center}
		\vspace{-.38cm}
	\end{figure*}
}

\newcommand{\boxref}[1]{Figure~\ref{#1}}

\RequirePackage{mdframed}

\newenvironment{titlebox}[5]
{\mdfsetup{
		style=#2,
		innertopmargin=1.1\baselineskip,
		skipabove={\dimexpr0.7\baselineskip+\topskip\relax},
		skipbelow={1em},needspace=3\baselineskip,
		singleextra={\node[#3,right=10pt,overlay] at (P-|O){~{\sffamily\bfseries #1 }};},%
		firstextra={\node[#3,right=10pt,overlay] at (P-|O) {~{\sffamily\bfseries #1 }};},
		frametitleaboveskip=9em,
		innerrightmargin=5pt
	}
	\newcommand{\TitleCaption}{#4}
	\newcommand{\TitleLabel}{#5}
	\begin{mdframed}[font=\small]
		\setlist[itemize]{leftmargin=13pt}\setlist[enumerate]{leftmargin=13pt}\raggedright% 
	}
	{\end{mdframed}
	\vspace{-1em}
	{\captionof{figure}{\small \TitleCaption}\label{\TitleLabel}}
	\medskip
}

\tikzstyle{normal} = [thick, fill=white, text=black, draw, rounded corners, rectangle, minimum height=.7cm, inner sep=3pt]
\tikzstyle{gray} = [thick, fill=gray!90, text=white, rounded corners, rectangle, minimum height=.7cm, inner sep=3pt]

\mdfdefinestyle{commonbox}{%
	align=center, middlelinewidth=1.1pt,userdefinedwidth=\linewidth
	innerrightmargin=0pt,innerleftmargin=2pt,innertopmargin=5pt,
	splittopskip=15pt,splitbottomskip=15pt
}
\mdfdefinestyle{roundbox}{style=commonbox,roundcorner=5pt,userdefinedwidth=\linewidth}

\newenvironment{systembox*}[3]
{\begin{strip}
		\vspace{\baselineskip}\begin{titlebox}{Functionality \normalfont #1}{roundbox}{normal}{#2}{#3}}
		{\end{titlebox}
\end{strip}}

\newenvironment{gsystembox*}[3]
{\begin{strip}
		\vspace{\baselineskip}\begin{titlebox}{Global Functionality \normalfont #1}{roundbox}{normal}{#2}{#3}}
		{\end{titlebox}
\end{strip}}

\newenvironment{protocolbox*}[3]
{\begin{strip}
		\begin{titlebox}{Protocol \normalfont #1}{commonbox}{normal}{#2}{#3}}
		{\end{titlebox}
\end{strip}}

\newenvironment{algobox*}[3]
{\begin{strip}
		\begin{titlebox}{Algorithm \normalfont #1}{commonbox}{normal}{#2}{#3}}
		{\end{titlebox}
\end{strip}}

\newenvironment{reductionbox*}[3]
{\begin{strip}
		\begin{titlebox}{Reduction \normalfont #1}{commonbox}{normal}{#2}{#3}}
		{\end{titlebox}
\end{strip}}

\newenvironment{gamebox*}[3]
{\begin{strip}
		\begin{titlebox}{Game \normalfont #1}{commonbox}{gray}{#2}{#3}}
		{\end{titlebox}
\end{strip}}

\newenvironment{simulatorbox*}[3]
{\begin{strip}
		\begin{titlebox}{Simulator \normalfont #1}{commonbox}{normal}{#2}{#3}}
		{\end{titlebox}
\end{strip}}

\mdfdefinestyle{mycommonbox}{%
	align=center, middlelinewidth=1.1pt,userdefinedwidth=\linewidth
	innerrightmargin=0pt,innerleftmargin=2pt,innertopmargin=3pt,
	splittopskip=7pt,splitbottomskip=3pt
}
\mdfdefinestyle{myroundbox}{style=mycommonbox,roundcorner=5pt,userdefinedwidth=\linewidth}

\newenvironment{titlebox*}[5]
{\mdfsetup{
		style=#2,
		innertopmargin=0.3\baselineskip,
		skipabove={0.4em},
		skipbelow={1em},needspace=3\baselineskip,
		frametitleaboveskip=5em,
		innerrightmargin=5pt
	}
	\newcommand{\TitleCaption}{#4}
	\newcommand{\TitleLabel}{#5}
	\begin{mdframed}[font=\small]
		\setlist[itemize]{leftmargin=13pt}\setlist[enumerate]{leftmargin=13pt}\raggedright% 
	}
	{\end{mdframed}
	\vspace{-2em}
	{\captionof{figure}{\normalfont \TitleCaption}\label{\TitleLabel}}
	\medskip
}

\newenvironment{mysystembox}[3]
{\begin{titlebox*}{Functionality \normalfont #1}{myroundbox}{normal}{#2}{#3}}
	{\end{titlebox*}}

\newenvironment{myprotocolbox}[3]
{\begin{titlebox*}{Protocol \normalfont #1}{mycommonbox}{normal}{#2}{#3}}
	{\end{titlebox*}}

\newenvironment{mysimulatorbox}[3]
{\begin{titlebox*}{Simulator \normalfont #1}{mycommonbox}{normal}{#2}{#3}}
	{\end{titlebox*}}
\newenvironment{mysystembox*}[3]
{\begin{strip}
		\vspace{\baselineskip}\begin{titlebox*}{Functionality \normalfont #1}{myroundbox}{normal}{#2}{#3}}
		{\end{titlebox*}
\end{strip}}

\newenvironment{mygsystembox*}[3]
{\begin{strip}
		\vspace{\baselineskip}\begin{titlebox*}{Global Functionality \normalfont #1}{myroundbox}{normal}{#2}{#3}}
		{\end{titlebox*}
\end{strip}}

\newenvironment{myprotocolbox*}[3]
{\begin{strip}
		\begin{titlebox*}{Protocol \normalfont #1}{mycommonbox}{normal}{#2}{#3}}
		{\end{titlebox*}
\end{strip}}

\newenvironment{myalgobox*}[3]
{\begin{strip}
		\begin{titlebox*}{Algorithm \normalfont #1}{mycommonbox}{normal}{#2}{#3}}
		{\end{titlebox*}
\end{strip}}

\newenvironment{myreductionbox*}[3]
{\begin{strip}
		\begin{titlebox*}{Reduction \normalfont #1}{mycommonbox}{normal}{#2}{#3}}
		{\end{titlebox*}
\end{strip}}

\newenvironment{mygamebox*}[3]
{\begin{strip}
		\begin{titlebox*}{Game \normalfont #1}{mycommonbox}{gray}{#2}{#3}}
		{\end{titlebox*}
\end{strip}}

\newenvironment{mysimulatorbox*}[3]
{\begin{strip}
		\begin{titlebox*}{Simulator \normalfont #1}{mycommonbox}{normal}{#2}{#3}}
		{\end{titlebox*}
\end{strip}}

\newcommand{\algoHead}[1]{\vspace{0.2em} \underline{\textbf{#1}} \vspace{0.3em}}

\makeatletter
\algnewcommand{\ExtendedState}[1]{\State
	\parbox[t]{\dimexpr\linewidth-\ALG@thistlm}{\hangindent=\algorithmicindent\strut\hangafter=3#1\strut}}
\makeatother

\algnewcommand\algorithmicinput{\textbf{Input:}}
\algnewcommand\Input{\item[\algorithmicinput]}

\algrenewcommand{\algorithmiccomment}[1]{{\color{gray}// #1}}

\let\emptyset\varnothing

\newcommand{\ckt}{\ensuremath{\mathsf{ckt}}}
\newcommand{\gate}{\mathsf{g}}
\newcommand{\wire}{\mathsf{w}}
\newcommand{\PS}{\ensuremath{\mathsf{A}}} 
\newcommand{\IS}{\ensuremath{\mathsf{I}}}
\newcommand{\OS}{\ensuremath{\mathsf{O}}}  
\newcommand{\MS}{\ensuremath{\mathsf{M}}} 
\newcommand{\DF}{\ensuremath{\mathsf{D}}}

\newcommand{\wx}{x}
\newcommand{\wy}{y}
\newcommand{\wz}{z}

\newcommand{\Wx}{\wire_{\wx}}
\newcommand{\Wy}{\wire_{\wy}}
\newcommand{\Wz}{\wire_{\wz}}
\newcommand{\Wxyz}{\ensuremath{\Wx, \Wy, \Wz}}
\newcommand{\Wxyzj}{\ensuremath{{\Wx}_j, {\Wy}_j, {\Wz}_j}}

\newcommand{\MPC}{\ensuremath{\mbox{MPC}}}
\newcommand{\negl}{\ensuremath{\mathsf{negl}}}
\newcommand{\csec}{\kappa}
\newcommand{\sparam}{\ensuremath{s}}
\newcommand{\abort}{\ensuremath{\mathtt{abort}}}
\newcommand{\continue}{\ensuremath{\mathtt{continue}}}
\newcommand{\flag}{\ensuremath{\mathsf{flag}}}

\newcommand{\Partyset}{\ensuremath{\mathcal{P}}}

\newcommand{\semi}{\ensuremath{\mathsf{s}}}
\newcommand{\mal}{\ensuremath{\mathsf{m}}}

\newcommand{\SELECT}{\ensuremath{\mathsf{select}}}
\newcommand{\INPUT}{\ensuremath{\mathsf{Input}}}
\newcommand{\OUTPUT}{\ensuremath{\mathsf{Output}}}

\newcommand{\Adv}{\ensuremath{\mathcal{A}}}
\newcommand{\Sim}{\ensuremath{\mathcal{S}}}

\newcommand{\Hash}{\ensuremath{\mathsf{H}}}

\newcommand{\rtt}{\ensuremath{\mathsf{rtt}}}

\newcommand{\ESet}{\ensuremath{P_1, P_2}}		%Evaluator
\newcommand{\EInSet}{\ensuremath{ \{1, 2\}}}	%Evaluator Index Set
\newcommand{\PInSet}{\ensuremath{ \{0, 1, 2\}}}	%Party Index Set

\newcommand{\val}{\ensuremath{\mathsf{v}}} 
\newcommand{\pad}{\ensuremath{\mathsf{\lambda}}} 
\newcommand{\mask}{\ensuremath{\mathsf{m}}}
\newcommand{\mpad}{\ensuremath{\mathsf{\delta}}} 
\newcommand{\sqd}{\ensuremath{[\cdot]}}
\newcommand{\sqr}[1]{[#1]}

\newcommand{\shrd}{\ensuremath{\llbracket \cdot \rrbracket}}
\newcommand{\shr}[1]{\ensuremath{\llbracket #1 \rrbracket}}
\newcommand{\Mask}[1]{\ensuremath{\mask_{#1}}}
\newcommand{\Mstar}[1]{\ensuremath{\mask_{#1}^{\star}}}
\newcommand{\Mbar}[1]{\ensuremath{\mask_{#1}^*}}
\newcommand{\Pad}[1]{\ensuremath{\pad_{#1}}}
\newcommand{\PadA}[1]{\ensuremath{\pad_{#1,1}}}
\newcommand{\PadB}[1]{\ensuremath{\pad_{#1,2}}}

\newcommand{\GammaV}[1]{\ensuremath{\gamma_{#1}}}
\newcommand{\GammaA}[1]{\ensuremath{\gamma_{#1,1}}}
\newcommand{\GammaB}[1]{\ensuremath{\gamma_{#1,2}}}

\newcommand{\Gammaxy}{\ensuremath{\gamma_{\wx \wy}}}
\newcommand{\GammaxyA}{\ensuremath{\gamma_{\wx \wy,1}}}
\newcommand{\GammaxyB}{\ensuremath{\gamma_{\wx \wy,2}}}

\newcommand{\Gammaxyj}{\ensuremath{\gamma_{{\wx}_j {\wy}_j}}}
\newcommand{\GammaxyjA}{\ensuremath{\gamma_{{\wx}_j {\wy}_j,1}}}
\newcommand{\GammaxyjB}{\ensuremath{\gamma_{{\wx}_j {\wy}_j,2}}}

\newcommand{\MPad}[1]{\ensuremath{\mpad_{#1}}}
\newcommand{\shrB}[1]{\ensuremath{{\llbracket #1 \rrbracket}^{\bf B}}}

\newcommand{\wma}{\mathsf{a}}
\newcommand{\wmb}{\mathsf{b}}
\newcommand{\wmc}{\mathsf{c}}

\newcommand{\wmd}{\mathsf{d}}
\newcommand{\wme}{\mathsf{e}}
\newcommand{\wmf}{\mathsf{f}}

\newcommand{\PiSemi}{\ensuremath{\Pi_{\mathsf{3pc}}^{\semi}}}
\newcommand{\PiShS}{\ensuremath{\Pi_{\mathsf{Sh}}^{\semi}}}
\newcommand{\PiAdd}{\ensuremath{\Pi_{\mathsf{Add}}}}
\newcommand{\PiMult}{\ensuremath{\Pi_{\mathsf{Mul}}}}
\newcommand{\PiMultS}{\ensuremath{\Pi_{\mathsf{Mul}}^{\semi}}}
\newcommand{\PiRecS}{\ensuremath{\Pi_{\mathsf{Rec}}^{\semi}}}

\newcommand{\PiMal}{\ensuremath{\Pi_{\mathsf{3pc}}^{\mal}}}
\newcommand{\PiShM}{\ensuremath{\Pi_{\mathsf{Sh}}^{\mal}}}
\newcommand{\PiMultM}{\ensuremath{\Pi_{\mathsf{Mul}}^{\mal}}}
\newcommand{\PiRecM}{\ensuremath{\Pi_{\mathsf{Rec}}^{\mal}}}
\newcommand{\PiFRec}{\ensuremath{\Pi_{\mathsf{fRec}}}}

\newcommand{\PiTripCheck}{\ensuremath{\Pi_\mathsf{prc}}}
\newcommand{\PiTripGen}{\ensuremath{\Pi_\mathsf{trip}}}
\newcommand{\PiRand}{\ensuremath{\Pi_{\mathsf{rand}}}}
\newcommand{\FTRIPLES}{\ensuremath{\mathcal{F}_{\mathsf{trip}}}}
\newcommand{\FTHREEMPC}{\ensuremath{\mathcal{F}_{\mathsf{3pc}}}}
\newcommand{\FSETUP}{\ensuremath{\mathcal{F}_{\mathsf{setup}}}}
\newcommand{\FTHREEMPCABORT}{\ensuremath{\mathcal{F}_{\mathsf{3pc}}^{\mathsf{Abort}}}}
\newcommand{\SPiSemi}{\ensuremath{{\mathcal S}_{\mathsf{3pc}}^{\semi}}}

\newcommand{\SPiMal}{\ensuremath{{\mathcal S}_{\mathsf{3pc}}^{\mal}}}
\newcommand{\SSetupMal}{\ensuremath{{\mathcal S}_{\mathsf{setup}}^{\mal}}}

\newcommand{\Va}{\mathsf{a}}
\newcommand{\Vb}{\mathsf{b}}

\newcommand{\Vp}{\mathsf{p}}

\newcommand{\Vq}{\mathsf{q}}
\newcommand{\Vr}{\mathsf{u}}

\newcommand{\Vu}{\mathsf{u}}

\newcommand{\VPad}[1]{\ensuremath{\overrightarrow{\pad_{#1}}}}
\newcommand{\VMask}[1]{\ensuremath{\overrightarrow{\mask_{#1}}}}

\newcommand{\MA}{\ensuremath{\vec{\mathbf{a}}}}
\newcommand{\MB}{\ensuremath{\vec{\mathbf{b}}}}

\newcommand{\MP}{\ensuremath{\vec{\mathbf{p}}}}
\newcommand{\MQ}{\ensuremath{\vec{\mathbf{q}}}}

\newcommand{\MW}{\ensuremath{\vec{\mathbf{w}}}}
\newcommand{\MX}{\ensuremath{\vec{\mathbf{x}}}}
\newcommand{\MZ}{\ensuremath{\vec{\mathbf{z}}}}

\newcommand{\VPQ}{\ensuremath{\mathbf{p}\mathbf{q}}}
\newcommand{\sval}{\mathsf{a}}
\newcommand{\msp}{\mathsf{p}}
\newcommand{\msq}{\mathsf{q}}
\newcommand{\sr}{\mathsf{r}}
\newcommand{\su}{\mathsf{u}}
\newcommand{\sv}{\mathsf{v}}

\newcommand{\MSB}[1]{\mathsf{msb}(#1)}

\newcommand{\Model}{\ensuremath{\mathtt{M}}}
\newcommand{\Client}{\ensuremath{\mathtt{C}}}

\newcommand{\sign}{\ensuremath{\mathsf{sign}}}

\newcommand{\BitExt}{\ensuremath{\mathsf{BitExt}}}

\newcommand{\piBitExtS}{\ensuremath{\Pi^{\semi}_{\BitExt}}}
\newcommand{\PiBitExtS}[2]{\ensuremath{\Pi^{\semi}_{\BitExt}({#1}, {#2})}}
\newcommand{\piBitExtM}{\ensuremath{\Pi^{\mal}_{\BitExt}}}
\newcommand{\PiBitExtM}[2]{\ensuremath{\Pi^{\mal}_{\BitExt}({#1}, {#2})}}

\newcommand{\funcML}[1]{\ensuremath{f_{\mathsf{#1}}}}
\newcommand{\PiDot}{\ensuremath{\Pi_{\mathsf{dp}}}}
\newcommand{\PiDotSemi}{\ensuremath{\Pi_{\mathsf{dp}}^{\semi}}}
\newcommand{\PiDotMal}{\ensuremath{\Pi_{\mathsf{dp}}^{\mal}}}

\newcommand{\share}[1]{\ensuremath{\llbracket#1\rrbracket}}

\newcommand{{\piaSh}}[1]{\ensuremath{\Pi^{#1}_{\Sh}}}

\newcommand{\shareB}[1]{\ensuremath{{\llbracket#1\rrbracket}^{\bf B}}}

\begin{document}
\sloppy
\date{}
\fancyhead{}
\title{ASTRA: High Throughput 3PC over Rings with Application to Secure Prediction}
\titlenote{This article is the full and extended version of an earlier article to appear in ACM CCSW 2019}

%-----------------------------------------------------------------------------------------------------------------------------------------------
\author{Harsh Chaudhari}
\affiliation{\institution{Indian Institute of Science, Bangalore India}}
\email{chaudharim@iisc.ac.in}

\author{Ashish Choudhury}
\authornote{This Publication is an
	outcome of the R\&D work undertaken in the project under the Visvesvaraya
	PhD  Scheme of  Ministry of Electronics \&
	Information  Technology,
	Government of India,
	being implemented by Digital India Corporation (formerly Media Lab Asia)}
\affiliation{\institution{International Institute of Information Technology Bangalore, India}}
\email{ashish.choudhury@iiitb.ac.in}

\author{Arpita Patra}
\authornote{Arpita Patra  would like to acknowledge financial support by Tata Trust Travel Grant 2019 and SERB Women Excellence Award 2017 (DSTO 1706).}
\affiliation{\institution{Indian Institute of Science, Bangalore India}}
\email{arpita@iisc.ac.in}

\author{Ajith Suresh}
\affiliation{\institution{Indian Institute of Science, Bangalore India}}
\email{ajith@iisc.ac.in}
%-----------------------------------------------------------------------------------------------------------------------------------------------

\maketitle
%-----------------------------------------------------------------------------------------------------------------------------------------------
\subsection*{Abstract}
%-----------------------------------------------------------------------------------------------------------------------------------------------
The concrete efficiency of secure computation has been the focus of many recent works. In this work, we present concretely-efficient  protocols for secure $3$-party computation (3PC) over a ring of integers modulo $2^{\ell}$  tolerating one corruption, both with semi-honest and malicious security.  Owing to  the fact that computation over ring emulates computation over the real-world system architectures, secure computation over  ring has gained momentum of late.  

Cast in the offline-online paradigm, our constructions present the most efficient online phase in concrete terms. In the semi-honest setting, our protocol  requires communication of $2$ ring elements  per multiplication gate during the  {\it online} phase, attaining a per-party cost of {\em less than one element}. This  is achieved for the first time in the regime of 3PC. In the {\it malicious} setting, our protocol requires communication of $4$  elements per multiplication gate during the online phase, beating the state-of-the-art protocol by $5$ elements. Realized with both the security notions of selective abort and  fairness, the malicious protocol with fairness  involves slightly more communication than its counterpart with abort security for the output gates {\em alone}.

We apply our techniques from $3$PC in the regime of secure server-aided  machine-learning (ML) inference for a range of prediction functions--  linear regression,  linear SVM regression,  logistic regression,  and linear SVM classification. Our setting considers a model-owner with trained model parameters and a client with a query, with the latter willing to learn the prediction of her query based on the model parameters of the former. The inputs and computation are outsourced to a set of three non-colluding servers. Our constructions catering to both semi-honest and the malicious world, invariably perform better than the existing constructions. 
%----
\section{Introduction}
\label{sec:intro}
Secure Multi-Party Computation (MPC) \cite{Yao82,GMW87,BGW88},   the holy grail of secure distributed computing,  enables a set of $n$ mutually distrusting parties to perform joint computation on their private inputs, in a way that no coalition of $t$ parties can learn more information  than the output (privacy) or affect the true output of the computation (correctness).  While MPC, in general, has been a subject of extensive research, the area of $\MPC$ with a small number of parties in the {\it honest majority} setting \cite{MRZ15,AFLNO16,FLNW17,ChandranGMV17,ByaliJPR18} has drawn popularity of late mainly due to its efficiency and simplicity. Furthermore, most real-time applications involve a small number of parties. Applications such as statistical and financial data analysis \cite{BogdanovTW12}, email-filtering \cite{LaunchburyADM14}, distributed credential encryption \cite{MRZ15}, Danish sugar beet auction \cite{BogetoftCDGJKNNNPST09} involve 3 parties. Well-known  MPC frameworks such as VIFF \cite{Gei07}, Sharemind \cite{BogdanovLW08} have been explored with 3 parties. Recent advances in secure machine learning  (ML) based on  MPC  have shown  applications with a small number of parties \cite{MohasselZ17, MakriRSV17,RiaziWTS0K18, MR18, WaghGC18}. MPC with a small number of parties helps  solve MPC over large population as well via server-aided computation, where a small number of servers jointly hold the input data of the large population and run an MPC protocol evaluating the desired function.

With motivations galore, the specific problem of three-party computation (3PC) tolerating one corruption has received phenomenal attention of late   \cite{AFLNO16,FLNW17,ABFLLNOWW17,LN17,CGHIKLN18,NV18,MRZ15,IshaiKKP15,PatraR18,ByaliJPR18,NV18}.  Leveraging honest majority,  this setting allows to attain stronger security goals such as \textit{fairness} (corrupt party receives the output only if all honest parties receive output) which are otherwise impossible with dishonest-majority \cite{Cleve86}. In this work, we revisit  the concrete efficiency of  3PC and to be specific, the efficiency of the input-dependent computation. 

The two typical lines of constructions that the regime  of MPC over small population offer are-- high-throughput~\cite{AFLNO16,FLNW17,ArakiBFLNO16,ABFLLNOWW17,CGHIKLN18,NV18}, and low-latency~\cite{IshaiKKP15,GordonR018,MRZ15, ChandranGMV17, ByaliJPR18, PatraR18} protocols. Relying on secret sharing mechanism,  the former category requires  low communication overhead (bandwidth) and simple computations. Catering to low-latency networks, this category takes a number of communication rounds proportional to the multiplicative depth of the circuit representing the function to be computed. On the other hand, the other category, relying on garbled circuits, requires a constant number of communication rounds and  serve better in high-latency networks such as the Internet.  The focus of this work is high-throughput 3PC.

Almost all high-throughput protocols evaluate a circuit that represents the function $f$ to be computed in a  secret-shared fashion. Informally,  the parties jointly maintain the invariant that for each wire in the circuit, the exact value over that wire is available in a secret-shared fashion among the parties, in a way that the adversary learns no information about the exact value from the shares of the corrupt parties. Upon completion of  the circuit evaluation, the parties jointly reconstruct the secret-shared function output. Intuitively, the security holds as no intermediate value is revealed during the computation. The deployed secret-sharing schemes  are typically linear, ensuring non-interactive evaluation of the linear gates. The communication is required {\em only} for the non-linear  (i.e.multiplication) gates in the circuit. The focus then turns on  improving the communication overhead per multiplication gate.  Recent literature has seen a  range of customized linear secret-sharing schemes over a small number of parties, boosting the performance for multiplication gate spectacularly \cite{FLNW17,ABFLLNOWW17,GordonR018}.

In an interesting direction towards improving efficiency,  MPC protocols are suggested  to be cast in two phases-- an offline phase that performs  {\it input-independent}  computation and an online phase that performs {\em fast input-dependent} computation utilizing the offline computation \cite{Bea91}.  The offline phase, run  in advance,  generates `raw material' in  a relatively expensive way to yield a  blazing-fast online phase.  This is very useful in a scenario where a set of parties agreed to  perform a specific computation repetitively over a period of time.  The parties can batch together the offline computations and generate a large volume of offline data to support the execution of multiple online phases. Popularly referred as offline-online paradigm \cite{Bea91}, there are constructions abound that show effectiveness of this paradigm  both in the theoretical \cite{Bea91,Bea95,BH06,BH08,BFO12,CP17}  and practical  \cite{DPSZ12,SPDZ2,SPDZ3,KOS16,BaumDTZ16,DamgardOS17,CramerDESX18,RiaziWTS0K18,KellerPR18}  regime. 

In yet another direction to improve practical efficiency,  secure computation for arithmetic circuits  over rings has gained momentum of late, while traditionally fields  have been the default choice.  Computation over rings models computation in the real-life computer architectures such as computation over CPU words of size 32 or 64 bits. In 3PC setting, the work of \cite{BogdanovLW08}  supports arithmetic circuits over arbitrary rings  with passive security, while   \cite{ABFLLNOWW17} offers active security.  The works of \cite{DamgardOS17, EeriksonOPPS19} improve online communication over arbitrary rings with active security, yet fall back to computation over large prime-order fields  in the offline phase. This forces the developer to depend on external libraries for fields (which are $10\times$-$100\times$ slower) compared to the real-world system architectures based on 32-bit and 64-bit rings.

\subsection{Our Contribution}
%----------------------------------------------------------------
In this work, we follow the offline-online paradigm and propose 3PC constructions over a ring $\Z{\ell}$ (that include Boolean ring $\Z{1}$) with  the most efficient online phase in concrete terms. Though the focus lies on the online phase,   the cost of offline phase is respected and is kept in check.  We present a range of constructions satisfying semi-honest and malicious security. We apply our techniques for secure prediction for a range of prediction functions in the outsourced setting and build a number of constructions tolerating semi-honest and malicious adversary. A common feature that all our constructions exude is that function-dependent communication is needed amongst {\em fewer} than three pairs in the online phase, yielding better online performance. We elaborate on our contributions.
%----------------------------------------------------------------
\paragraph{Secure 3PC}
%----------------------------------------------------------------
Our 3PC protocol with semi-honest security requires communication of two  elements per multiplication during the online phase.  The per-party online cost of our protocol is less than one element per multiplication, a property achieved for the first time in the 3PC setting. This improvement comes from the use of a form of linear secret-sharing scheme inspired from the work of \cite{GordonR018} that allows offloading the task of one of the parties in the offline phase and requires  {\em only} two parties to talk to each other in the online phase. This essentially implies that the evaluation of multiplication gates in the online phase requires the presence of just two parties, unlike the  previous protocols \cite{AFLNO16,FLNW17,ABFLLNOWW17,LN17,CGHIKLN18} that insist all the three parties be awake throughout the computation. One exception is the case of Chameleon \cite{RiaziWTS0K18}, where two parties perform the online computation with the help of correlated randomness generated by a {\it semi-trusted} party in the offline phase. Though the model looks similar in the semi-honest setting, we achieve a stronger security guarantee by allowing the third party to be maliciously corrupted. Moreover, our multiplication protocol in the semi-honest setting requires an online communication of 2 ring elements as opposed to 4 of \cite{RiaziWTS0K18}. We achieve this $2\times$ improvement while maintaining the same offline cost (1 element) of \cite{RiaziWTS0K18}.

For the malicious case, our protocol requires a {\it total} communication of four elements per multiplication during the online phase. The state-of-the-art protocol over {\it rings} requires nine ring elements per multiplication in the online phase. Lastly, we  boost the security of our malicious protocol to fairness  without affecting its cost  per multiplication. The inflation inflicted is purely for the output gates and to be specific for output reconstruction. The key contribution of the fair protocol lies in constructing a fair reconstruction protocol that ensures a corrupt party receives the output if and only if the honest parties receive.  The fair reconstruction does not resort to a broadcast channel and instead rely on a  new concept of `proof of origin'  that tackles the confusion a sender can infuse in the absence of broadcast channel by sending different messages to its fellow parties over private channels. 

In \tabref{concrete3PC}, we compare our work with the most relevant works. The communication specifies the number of bits that needs to be communicated per multiplication gate in the amortized sense.
%-----
\begin{table}[htb!]
	\centering
	\resizebox{0.46\textwidth}{!}
	{
		\begin{tabu} to 1\textwidth {  l | l | l | l | l | l | l}
			\toprule
			\multicolumn{3}{c|}{Semi-honest} & \multicolumn{4}{c}{Malicious}  \\ 	
			\midrule
			Ref. &  Offline & Online & Ref.  & Offline & Online & Fair?\\ 
			\midrule     
			%---
			\cite{AFLNO16} 			& $0$        & $3  \ell$  
			& \cite{ABFLLNOWW17}	& $12  \ell$ & $9 \ell$ 	& \xmark \\ \midrule
			{\bf This} 				& $ \ell$ 	 & $\mathbf{2  \ell}$  
			& {\bf This}			& $21 \ell$  & $\mathbf{4  \ell}$ & \cmark \\
			%---	
			\bottomrule
			%-----
		\end{tabu}
	}
	\caption{Concrete Comparison of our 3PC protocols\label{tab:concrete3PC}}
\end{table}
%------
%----------------------------------------------------------------
\paragraph{Secure ML Prediction}\label{sec:privMLPrelims}
%----------------------------------------------------------------
The growing influx of data makes ML a promising applied science, touching human life like never before. Its potential can be leveraged to  advance areas such as medicine \cite{EstevaKNKSBT17}, facial recognition \cite{SchroffKP15}, banking, recommendation services, threat analysis, and authentication technologies.  Many technology giants  such as Amazon, Microsoft, Google, Apple are offering cloud-based ML services to their customers both in the form of training platforms that train models on customer data and pre-trained models that can be used for inference, often referred as `ML as a Service (MLaaS)'.  However, these huge promises can only be unleashed when  rightful privacy concerns, due to ethical, legal or competitive reasons,  can be brought to control via privacy-preserving techniques. This is when privacy-preserving techniques such as MPC meets ML, with the former serving extensively in an effective way both for secure training and prediction \cite{MohasselZ17, RiaziWTS0K18, MR18, WaghGC18, LiuJLA17, LaurLM06, Dahl18}. This has a huge impact on the efficiency 

In this work, we target secure prediction  where a model-owner  holding the model parameters enables a client to receive a prediction result to its query as per the model, respecting  privacy concerns of each other. Following the works of \cite{MohasselZ17, MakriRSV17,RiaziWTS0K18, MR18, WaghGC18}, we envision a server-aided setting where the inputs and computation are outsourced to a set of servers. We consider some of the widely used ML algorithms, namely linear regression and linear support vector machines (SVM) regression for {\em regression} task and  logistic regression and SVM classification for {\em classification} task \cite{DRHPSG2000,Bishop06}. We propose an efficient protocol for {\em secure comparison}  that is an important building block for  classification task. We exploit the asymmetry in our secret sharing scheme and forgo expensive primitives such as garbled circuits or parallel prefix adders, which are used in \cite{MohasselZ17} and \cite{MR18}. As emphasized below, our technique allows attaining a constant round complexity for classification tasks.

In \tabref{concreteML}, we compare our results with the best-known construction  of  ABY3 \cite{MR18} that uses 3-server setting.   As the main focus of ABY3 is training, they develop an efficient technique for fixed-point multiplication in shared fashion, tackling the overflow and  accuracy issues in the face of repeated multiplications. Such techniques can be avoided for functions inducing circuit of multiplicative depth one.  Hence we compare with the version of ABY3 that skips this and present below a consolidated comparison in terms of communication.  Following the works in the domain of server-aided prediction, we only count the cost incurred by the servers to compute the output in shared form from the inputs in shared form, ignoring the cost  for sharing the inputs and reconstructing the output.  `Reg' denotes regression, `Class' denotes classification and `Round' denotes the number of online rounds. Here $\ell$ denotes the size of the underlying ring $\Z{\ell}$ (in bits) and $d$ denotes  the number of features. 
%--------
\begin{table}[htb!]
	\centering
	\resizebox{0.46\textwidth}{!}
	{
		\begin{tabu} to 1\textwidth { c | c | c | c | c | c }
			\toprule
			\multirow{2}[2]{*}{Ref.} 	& \multirow{2}[2]{*}{Param.} & \multicolumn{2}{c|}{Semi-honest} & \multicolumn{2}{c}{Malicious}  \\
			\cmidrule{3-6}
			&      & Reg  & Class & Reg & Class \\
			\midrule
			\multirow{3}[2]{*}{ABY3}
			& Offline & $0$ 	 & $0$      
			& $12 d \ell$ 
			& $12 d \ell + 24 \ell$ \\ \cmidrule{2-6}
			& Online  & $3 \ell$ & $9 \ell$ 
			& $9 d \ell$   
			& $9 d \ell + 18 \ell$\\ \cmidrule{2-6}
			& Round  & $1$ & $\log \ell + 1$ & $1$ & $\log \ell + 1$ \\
			\midrule
			\multirow{3}[2]{*}{\bf This}
			& Offline 	& $\ell$  				& $\ell$ 
			& $21 d \ell$ 			& $21 d \ell + 46 \ell$\\ \cmidrule{2-6}
			& Online  & $\mathbf{2 \ell}$ 	& $\mathbf{4 \ell + 2}$
			& $\mathbf{2 d \ell + 2 \ell}$ 	& $\mathbf{2 d \ell + 8 \ell + 1}$\\ \cmidrule{2-6}
			& Round   & $1$ & $\mathbf{3}$ & $1$ & $\mathbf{4}$ \\
			%---
			\bottomrule
			%-----
		\end{tabu}
	}
	\caption{Concrete Comparison of Our ML Protocols\label{tab:concreteML}}
\end{table}
%---------
The values in \tabref{concreteML} indicate that our protocol clearly outperforms ABY3, in terms of online communication in all the settings. In the semi-honest setting, this is achieved since we are able to shift $33\%$ of the overall communication to the offline phase. In the malicious setting, online communication is further improved because of our efficient dot-product protocol. Moreover, our novel construction for secure comparison allows the classification protocols to be {\em round constant} unlike ABY3 which requires $\log{\ell} + 1$ rounds. 

%----------------------------------------------------------------
\paragraph{Implementation}
%----------------------------------------------------------------
For 3PC, we implement our protocols over a ring $\Z{32}$  and compare with the state-of-the-art protocols, namely \cite{AFLNO16}  in the semi-honest setting and \cite{ABFLLNOWW17} in the malicious setting. We use latency (runtime) and online throughput as the parameters for the comparison. The online throughput in LAN setting is computed as the number of AES circuits computed per second in the online phase. As an AES circuit requires more than a second  in WAN setting, we take a different measure which is  the number of AND gates per second. We observe that our protocols improve the online throughput of the existing one by a factor of $1.05\times$ to $1.51\times$ over various settings. For the WAN setting,  this improvement translates to computing {\em additional} AND gates of the range  $1.44$  to $4.39$ millions per second.

For secure prediction, we implement our work using MNIST~\cite{MNIST10} dataset where $d=784$ and with $\ell= 64$  in both LAN and WAN setting. We observe an improvement of $1.02\times$ to $2.56\times$ over ABY3 \cite{MR18}, in terms of online throughput, over various settings for regression algorithms. For classification algorithms, the improvement ranges from $1.5\times$ to $2.93\times$.

%----------------------------------------------------------------
%\paragraph{Open Problems}
%----------------------------------------------------------------
%Our techniques are tailor-made for 3PC with 1 corruption. Extending these techniques to the case of an arbitrary $Q^{(2)}$ adversary structure~\cite{SmartW19} is left as an open problem.

%----
%\clearpage
\section{Preliminaries and Definitions}
\label{sec:prelim}

We consider a set of three parties $\Partyset = \{ P_0, P_1, P_2 \}$ that are connected by pair-wise private and authentic channels in a synchronous network. The function $f$ to be evaluated is expressed as a circuit $\ckt$ over an arithmetic ring $\Z{\ell}$, consisting of $2$-input addition and multiplication gates. The topology of the circuit is assumed to be publicly known. The term $\DF$ denotes the multiplicative depth of the circuit, while $\IS, \OS, \PS, \MS$ denote the number of input wires, output wires, addition gates and multiplication gates respectively in $\ckt$. We use the notation $\Wx$ to denote a wire $\wire$ with value $\wx$ flowing through it. We use $\gate = (\Wxyz)$ to denote a gate in the $\ckt$ with left input wire $\Wx$, right input wire $\Wy$ and output wire $\Wz$. In our protocols, we divide $\Partyset$ into disjoint sets $\{P_0 \}$ and $\{P_1, P_2 \}$, where $P_0$ acts as a ``distributor" to do the ``pre-processing" during the offline phase, which is utilized by the ``evaluators" $P_1, P_2$ to evaluate $\ckt$ during the online phase. We use the superscripts ``$\semi$" and ``$\mal$" to distinguish the protocols in the semi-honest and malicious setting respectively. The protocols over boolean ring $\Z{1}$ can be obtained by replacing the arithmetic operations addition ($+$) and multiplication ($\times$) with XOR ($\xor$) and AND ($\cdot$) respectively. Below, we present the tools needed for our protocol.
%----------------------------------------------------------------
\subsection{Collision Resistant Hash}
\label{sec:hash}
%----------------------------------------------------------------
Consider a hash function family $\Hash = \mathcal{K}\times \mathcal{L} \rightarrow \mathcal{Y}$. The hash function $\Hash$ is said to be collision resistant if for all probabilistic polynomial-time adversaries $\Adv$, given the description of $\Hash_k$ where {$k \in_R \mathcal{K}$}, there exists a negligible function $\negl()$ such that $\Pr[ (x_1,x_2) \leftarrow \Adv(k):(x_1 \ne x_2) \wedge \Hash_k(x_1)=\Hash_k(x_2)] \leq \negl(\csec)$, where $m = \mathsf{poly}(\csec)$ and $x_1,x_2 \in_R \{0,1\}^m$.

%----------------------------------------------------------------
\subsection{Shared Key Setup}
\label{sec:keysetup}
%----------------------------------------------------------------
To save communication between the parties, a one-time setup that establishes pre-shared random keys for a pseudo-random function (PRF) $F$ is used. A similar  setup has been used in the known protocols in the 3PC setting \cite{FLNW17, ABFLLNOWW17, MR18}.   Here $F : {0, 1}^{\csec} \times {0, 1}^{\csec} \rightarrow X$ is a secure PRF, with co-domain $X$ being $\Z{\ell}$. The set of keys are:
%-------------
\begin{myitemize}
	\item[--] One key shared between every pair-- $k_{01}, k_{02},\allowbreak k_{12}$ for the parties $(P_0, P_1), (P_0, P_2), (P_1, P_2)$ respectively.
	\item[--] One shared key amongst all-- $k_{\Partyset}$. 
\end{myitemize}
%--------------
If parties $P_0, P_1$ wish to sample a random value $r$ non-interactively, they invoke $F_{k_{01}}(id_{01})$ to obtain $r$, where $id_{01}$ is a counter that the parties update locally after every PRF invocation. The key used to sample a value will be clear from the context (from the identities of the pair that samples or from the fact that it is sampled by all) and will  be omitted. We model the key setup via a functionality $\FSETUP$ that  can be realized using any secure  MPC protocol. 
%----
\section{Sharing Semantics}
\label{sec:Semantics}
In this section, we explain two variants of secret sharing that are used in this work. Both the variants operate over arithmetic ($\Z{\ell}$) and boolean ($\Z{1}$) rings.
\vspace{1mm}
%----------------------------------------------------------------
\paragraph{$\sqd$-sharing}
%----------------------------------------------------------------
A value $\val$ is said to be $\sqd$-shared among parties $\ESet$, if the parties $P_1$ and $P_2$ respectively holds the values $\val_1$ and $\val_2$ such that $\val = \val_1 + \val_2$. We use $\sqd_{P_i}$ to denote the $\sqd$-share of  party $P_i$ for $i \in \EInSet$.
\vspace{1mm}
%----------------------------------------------------------------
\paragraph{$\shrd$-sharing}
%----------------------------------------------------------------
A value $\val$ is said to be $\shrd$-shared among parties $P_0, P_1$ and $P_2$, if 
\begin{myitemize}
	%---
	\item[--] there exists values $\Pad{\val}, \Mask{\val}$ such that $\val = \Mask{\val} - \Pad{\val}$.
	%---
	\item[--] $P_0$ holds   $\PadA{\val}$ and $\PadB{\val}$. 
	%---
	\item[--] $P_1$ and $P_2$ hold $(\Mask{\val}, \PadA{\val})$ and $(\Mask{\val}, \PadB{\val})$ respectively. 
	%---
\end{myitemize}
We denote $\shrd$-share of the parties as $\shr{\val}_{P_0} = (\PadA{\val}, \PadB{\val})$, $\shr{\val}_{P_1} = (\Mask{\val}, \PadA{\val})$ and $\shr{\val}_{P_2} = (\Mask{\val}, \PadB{\val})$. We use $\shr{\val} = (\Mask{\val}, \sqr{\Pad{\val}})$ to denote the $\shrd$-share of $\val$. 
\vspace{1mm}
%----------------------------------------------------------------
\paragraph{Linearity of the secret sharing schemes}
%----------------------------------------------------------------
Given the $\sqd$-sharing of $\wx, \wy \in \Z{\ell}$ and public constants $c_1, c_2 \in \Z{\ell}$, parties can locally compute $\sqr{c_1 \wx + c_2 \wy}$. To see this,
%------------
\begin{align*}
	\sqr{c_1 \wx + c_2 \wy} = (c_1 \wx_1 + c_2 \wy_1, c_1 \wx_2 + c_2 \wy_2) = c_1 \sqr{\wx} + c_2 \sqr{\wy}
\end{align*}
%------------
It is easy to see that the linearity trivially extends to $\shrd$-sharing as well. That is, given the $\shrd$-sharing of $\wx, \wy$ and public constants $c_1, c_2$, parties can locally compute $\shr{c_1 \wx + c_2 \wy}$.
%------------
\begin{align*}
	\shr{c_1 \wx + c_2 \wy} &= (c_1 \Mask{\wx} + c_2 \Mask{\wy}, c_1 \sqr{\Pad{\wx}} + c_2 \sqr{\Pad{\wy}})\\ &= c_1 \shr{\wx} + c_2 \shr{\wy}
\end{align*}
%------------
The linearity property enables parties to {\em locally} perform the operations such as addition and multiplication with a public constant.
%----
%\clearpage
\section{Our 3PC Protocol}
\label{sec:ThreePC}
We start with our 3PC protocol $\PiSemi$ that securely evaluates any arithmetic circuit over $\Z{\ell}$ for $\ell \geq 1$, tolerating semi-honest adversaries.
%----------------------------------------------------------------
\subsection{3PC with semi-honest security}
\label{sec:3pcsemi}
%----------------------------------------------------------------
Our  protocol $\PiSemi$ has three stages-- input-sharing, circuit-evaluation, and output-reconstruction. During input-sharing stage, each party generates a random $\shrd$-sharing of its input. During the circuit-evaluation stage, the parties evaluate $\ckt$ in a $\shrd$-shared fashion. During the output-reconstruction stage, the parties reconstruct the $\shrd$-shared circuit outputs. All the stages (except output-reconstruction) can be cast in the offline and online phase, where steps independent of the actual inputs can be executed in the offline phase. At a high level,  the $\sqd$-sharing needed behind every $\shrd$-shared value in the online phase is precomputed, while the $\shrd$-sharing of values themselves are computed in the online phase.  We distinguish these steps as {\em Offline} and {\em Online} steps respectively. While the {\em Offline}  steps are executed {\em only} by the distributor $P_0$,  the {\em Online} steps are executed {\em only} by the evaluators $P_1$ and $P_2$. We now individually elaborate on each of the stages. 

%----------------------------------------------------------------
\paragraph{Input-sharing Stage}
%----------------------------------------------------------------
Protocol $\PiShS(P_i, \wx)$ (\boxref{fig:PiShS}) allows party $P_i \in \Partyset$,  the designated party to give input $\wx \in \Z{\ell}$ to wire $\Wx$, to $\shrd$-share its input. In the offline step, parties locally sample $\PadA{\wx}$ and $\PadB{\wx}$ using their shared randomness such that parties $P_0$ and $P_i$ learns the entire $\Pad{\wx}$. In the online step, $P_i$ computes $\Mask{\wx}$ using $\lambda_\wx$ and sends it to the evaluators. 
%--------
\begin{myprotocolbox}{}{Protocol $\PiShS(P_i, \wx)$}{fig:PiShS}
	%-----
	\justify 
	\algoHead{Offline:}
	\begin{myitemize}
		%-----
		\item[--] If $P_i = P_0$, parties $P_0, P_j$ for $j \in \{1,2\}$ locally sample a random $\Pad{\wx, j} \in \Z{\ell}$.
		%-----
		\item[--] If $P_i = P_1$, parties $P_0, P_1$ sample a random $\PadA{\wx} \in \Z{\ell}$ while all the parties in $\Partyset$ sample a random $\PadB{\wx}$. % using $k_{\Partyset}$.
		%-----
		\item[--] If $P_i = P_2$, parties $P_0, P_2$ sample a random $\PadB{\wx} \in \Z{\ell}$ while all the parties in $\Partyset$ sample a random $\PadA{\wx}$. % using $k_{\Partyset}$. \commentA{what is the reason for using different common keys for these two cases?}\commentAJ{Not needed. Changed}
		%-----
	\end{myitemize}
	%
	%-----
	\justify 
	\algoHead{Online:} $P_i$ sends $\Mask{\wx} = \wx + \Pad{\wx}$ to every $P_j$ for $j \in \{1,2\}$ who then sets $\shr{\wx}_{P_j} = (\Mask{\wx}, \Pad{\wx, j})$.
	%------
\end{myprotocolbox}
%------

%----------------------------------------------------------------
\paragraph{Circuit-evaluation Stage}
%----------------------------------------------------------------
Here parties evaluate each gate $\gate$ in the $\ckt$ in the {\it topological} order, where they maintain the invariant that given inputs of $\gate$ in $\shrd$-shared fashion, parties generate $\shrd$-sharing for the output of $\gate$.  If $\gate$ is an addition gate $(\Wxyz)$, then this is done locally using the linearity of $\shrd$-sharing, as per the protocol $\PiAdd$ (\boxref{fig:PiAdd}).
%------
\begin{myprotocolbox}{}{Protocol $\PiAdd(\Wxyz)$}{fig:PiAdd}
	%----
	\justify
	\algoHead{Offline:}  $P_0, P_1$ set $\PadA{\wz} = \PadA{\wx} + \PadA{\wy}$, while $P_0, P_2$ set $\PadB{\wz} = \allowbreak \PadB{\wx} + \PadB{\wy}$.  
	%---
	
	\justify
	\algoHead{Online:} $P_1$ and $ P_2$ set $\Mask{\wz} = \Mask{\wx} + \Mask{\wy}$.   
	%---
	%---     
\end{myprotocolbox}
%------
If $\gate = (\Wxyz)$ is a multiplication gate, then given $\shr{\wx} = (\Mask{\wx}, \sqr{\Pad{\wx}})$ and $\shr{\wy} = (\Mask{\wy}, \sqr{\Pad{\wy}})$, the parties compute $\shr{\wz}$ by running the protocol $\PiMultS$ (\boxref{fig:PiMultS}). During the offline phase, parties generate $\Pad{\wz}$ for the gate output. In addition, $P_0$ also $\sqr{\cdot}$-shares the product of the masks of the gate inputs ($\Pad{\wx} \Pad{\wy}$), both of which are  known to $P_0$ as a part of  $\shr{\wx}_{P_0}$ and $\shr{\wy}_{P_0}$. Online phase is executed by $\{P_1, P_2 \}$, where they locally generate $\sqr{\Mask{\wz}}$, followed by reconstructing $\Mask{\wz}$.
%-----
\begin{myprotocolbox}{}{Protocol $\PiMultS(\Wxyz)$}{fig:PiMultS}
	%-----
	\justify
	\algoHead{Offline:}  
	%-------
	\begin{myitemize}
		%-----
		\item[--]  $P_0$ and $ P_1$ locally sample random $\PadA{\wz}, \GammaxyA \in \Z{\ell}$, while $P_0 $ and $ P_2$ locally sample a random $\PadB{\wz}$ .                  
		%----
		\item[--] $P_0$ computes $\Gammaxy = \Pad{\wx} \Pad{\wy}$ and sends $\GammaxyB = \Gammaxy - \GammaxyA$ to $P_2$.                  
		%-----
	\end{myitemize}
	%------     
	
	\justify
	\algoHead{Online:} 
	%------
	\begin{myitemize}  
		%----
		\item[--]  $P_i$ for $i \in \EInSet$ locally computes $\sqr{\Mask{\wz}}_{P_i} = (i-1)\Mask{\wx} \Mask{\wy} - \allowbreak \Mask{\wx} \sqr{\Pad{\wy}}_{P_i} - \Mask{\wy} \sqr{\Pad{\wx}}_{P_i} + \sqr{\Pad{\wz}}_{P_i} +  \sqr{\Gammaxy}_{P_i}$. 
		%---    
		\item[--] $\ESet$ mutually exchange their shares and reconstruct $\Mask{\wz}$.
		%---
	\end{myitemize}
	%----  
\end{myprotocolbox}
%-----

%----------------------------------------------------------------
\paragraph{ Output-reconstruction Stage}
%----------------------------------------------------------------
To reconstruct the output from $\shr{\wy}$, we observe that the missing share of party $P_i$, for $i \in \PInSet$, is held by the other two parties. Thus, one among the other two parties can send the missing share to $P_i$, who then computes the output as $\wy = \Mask{\wy} - \PadA{\wy} - \PadB{\wy}$. We call the resultant protocol as $\PiRecS$. 
%------

We combine the aforementioned stages and present $\PiSemi$ in \boxref{fig:PiSemi}. 
\begin{myprotocolbox}{}{The semi-honest 3PC protocol $\PiSemi$}{fig:PiSemi}
	\justify
	\algoHead{Pre-processing (Offline Phase):}
	\begin{myitemize}      
		%---         
		\item[--] {\em Input wires}: For $j = 1, \ldots, \IS$, corresponding to the circuit-input $x_j$, parties execute the offline steps of the instance
		$\PiShS(P_i, x_j)$.
		%---         
		\item[--] For each gate $\gate$ in $\ckt$ in the topological order,  execute the offline steps of the instance $\PiMultS(\Wxyzj)$  if $\gate$ is the  $j$th multiplication gate where $j \in \{1, \ldots, \MS \}$  or respectively the offline steps of the instance $\PiAdd(\Wxyzj)$ if $\gate$ is  the $j$th  addition gate where $j \in \{1, \ldots, \PS \}$.    \\
		%----  
	\end{myitemize}
	%-----
	
	\justify
	\algoHead{Circuit Evaluation (Online Phase):}
	%------------
	\begin{myitemize}
		%---         
		\item[--] {\em Sharing Circuit-input Values}: For $j = 1, \ldots, \IS$, corresponding to the circuit-input $x_j$, party $P_i$ executes the online steps of the instance
		$\PiShS(P_i, x_j)$, where $P_i$ is the party designated to provide $x_j$.
		%----
		\item[--] {\em Gate Evaluation}: For each gate $\gate$ in $\ckt$ in the topological order, $P_1, P_2$ execute the online steps of the instance $\PiMultS(\Wxyzj)$  if $\gate$ is the  $j$th multiplication gate where $j \in \{1, \ldots, \MS \}$  or respectively the online steps of the instance $\PiAdd(\Wxyzj)$ if $\gate$ is  the $j$th  addition gate  where $j \in \{1, \ldots, \PS \}$. 
		%----     
		\item[--] {\em Output Reconstruction}:  Let $\shr{y_1 }, \ldots, \shr{y_{\OS}}$ be the shared function outputs, where for $j = 1, \ldots, \OS$, we have $\shr{y_j}_{P_0} = \allowbreak \sqr{\Pad{y_j}}, \shr{y_j}_{P_1} = (\Mask{y_j}, \sqr{\Pad{y_j}}_{P_1})$ and $\shr{y_j}_{P_2} = (\Mask{y_j}, \sqr{\Pad{y_j}}_{P_2})$. The parties in $\Partyset$ reconstruct $y_j$ by executing the instance $\PiRecS(\shr{y_j}, \Partyset)$.
		%---
	\end{myitemize}
\end{myprotocolbox}
%------------------

\paragraph*{Correctness and Security} We prove correctness and argue security informally below. 
%--------------
\begin{theorem}[Correctness]
	Protocol $\PiSemi$ is correct.
\end{theorem}
%-----
\begin{proof}
	We claim that for every wire  in $\ckt$, the parties hold a $\shrd$-sharing of the wire value in $\PiSemi$. The correctness then follows from the fact that for the circuit-output wires, the corresponding $\shrd$-sharing is reconstructed correctly. The claim for circuit-input wires follows from $\PiShS$, while for addition gates it follows from the linearity of $\shrd$-sharing. Consider a multiplication gate $(\Wxyz)$, evaluated as per $\PiMultS$, where $\Mask{\wx} = \wx + \Pad{\wx} $, $\Mask{\wy} = \wy + \Pad{\wy}$ and $\Gammaxy = \Pad{\wx} \Pad{\wy}$. We argue that $\Mask{\wz}$ as computed in online step of $\PiMultS$ results in  $\wx \wy + \Pad{\wz}$ and hence  at the end of $\PiMultS$, the parties hold $\shr{\wz}$. This is because $\Mask{\wz} = \Mask{\wx}  \Mask{\wy} -  \Mask{\wx} \Pad{\wy} - \Mask{\wy} \Pad{\wx} + \Pad{\wz} + \Gammaxy = (\Mask{\wx} - \Pad{\wx})(\Mask{\wy} - \Pad{\wy}) + \Pad{\wz} = \wx \wy + \Pad{\wz}$. The linearity of $\sqd$-sharing implies that $P_1$ and $P_2$ correctly compute a $\sqd$-sharing of $\Mask{\wz}$.
\end{proof}
%-----

The security is  argued as follows. If $P_0$ is corrupt, then the security follows since $P_0$ never sees the masked values over the intermediate wires. If one of the evaluators is corrupt, then the security holds since the corrupt evaluator knows only one of the shares of the mask while the other share is picked at random. The detailed security proof appear in Appendix~\ref{app:3PCSemi} where we show our protocol emulates  the functionality $\FTHREEMPC$ for computing a $3$-party function $f$ in the semi-honest setting as given in \boxref{fig:FTHREEMPC}. 
%---------
\begin{mysystembox}{}{Functionality $\FTHREEMPC$}{fig:FTHREEMPC}
	\justify $\FTHREEMPC$ interacts with the parties in $\Partyset$ and the adversary $\Sim$ and is parameterized by a $3$-ary function $f$, represented by a publicly known arithmetic circuit $\ckt$ over $\Z{\ell}$.
	%---
	\justify Upon receiving the input $\wx_1, \ldots, \wx_{\IS}$ from the respective parties in $\Partyset$, where  each $\wx_i \in \Z{\ell}$, the functionality computes $(\wy_1, \ldots, \wy_{\OS}) = f(\wx_1, \newline \ldots, \wx_{\IS})$  and sends $\wy_1, \ldots, \wy_{\OS}$ to the parties in $\Partyset$.
\end{mysystembox}
%---------
%---------
\begin{theorem}
	\label{thm:PiSemiCost}
	%-----
	$\PiSemi$ requires one round with communication of $\MS$ ring elements during the offline phase. In the online phase, $\PiSemi$ requires one round with communication of at most $2\IS$ ring elements in the Input-sharing stage, $\DF$ rounds with communication of $2\MS$  ring elements for circuit-evaluation stage and one round with communication of $3\OS$  elements for the output-reconstruction stage.
	%-------
\end{theorem}
%-----------
\begin{proof}
	%-----
	During the offline phase, the $\sqd$-shares of every $\pad$ are generated non-interactively. For the multiplication gates, generating $\sqd$-sharing of $\Gammaxy$ values requires one round and communication of $\MS$ elements. During the online phase, generating the $\shrd$-sharing of circuit-inputs requires one round. For each input of $P_0$, generating the $\shrd$-sharing  requires a communication of $2$ elements, while the same for $P_1$/$P_2$ requires one element.  So, the Input-sharing phase needs one round and communication of at most $2\IS$ elements. Evaluating the addition gates is free, while the same for each multiplication gate requires one round and communication of $2$ elements to reconstruct the $\Mask{\wz}$ value. Hence the circuit-evaluation phase needs $\DF$ rounds and communication of $2\MS$ elements. Reconstructing the circuit-outputs require one round and communication of $3\OS$ elements. 
	%-----
\end{proof}

%----------------------------------------------------------------
\subsection{3PC with malicious security}
\label{sec:3pcmal}
%----------------------------------------------------------------
In this section, we describe our maliciously secure 3PC protocol $\PiMal$ that securely evaluates any arithmetic circuit over $\Z{\ell}$. Similar to $\PiSemi$, protocol $\PiMal$ has three stages-- input-sharing, circuit-evaluation and output-reconstruction. 

%----------------------------------------------------------------
\paragraph{Input Sharing and Output Reconstruction Stages}
We begin with the sharing and reconstruction protocols in the malicious setting, which can  readily replace $\PiShS$ and $\PiRecS$ in  $\PiMal$ to help obtain maliciously-secure  input sharing and output reconstruction stage.

In the malicious setting, we need to ensure that the shares possessed by the honest parties are {\em consistent}. By consistent shares, we mean that the common share possessed by the honest parties should be the same. In protocol $\PiShS$, the $\pad$-shares will be consistent since they are generated non-interactively. But, if a corrupt $P_0$ owns a value $\wx$ and wants to create an inconsistent $\shr{\wx}$-sharing, he can send two different versions of $\Mask{\wx}$ to $P_1$ and $P_2$. To detect this inconsistency, $P_1, P_2$ exchange $\Hash(\Mask{\wx})$ and $\abort$ if there is a mismatch. The parties can exchange a combined hash for all the wires where $P_0$ is the owner and thus the cost reduces to two hash values in the amortized sense. We call the resultant protocol as $\PiShM$.

For reconstruction, let $\shr{\wy}$ be a sharing to be reconstructed where $\shr{\wy}_{P_0} =  (\PadA{\wy}, \PadB{\wy})$, $\shr{\wy}_{P_1} = (\Mask{\wy}', \PadA{\wy}')$ and $\shr{\wy}_{P_2} = (\Mask{\wy}'', \PadB{\wy}')$ (the distinction in the notation is done to differentiate the shares held by each party). Protocol $\PiRecM(\shr{\wy}, \Partyset)$ (\boxref{fig:PiRecM}) enables each {\it honest} party in $\Partyset$ to either compute $\wy$ or output $\bot$. 
%-----
\begin{myprotocolbox}{}{Protocol $\PiRecM(\shr{\wy}, \Partyset)$}{fig:PiRecM}   
	
	\algoHead{Online:} 
	%------
	\begin{myitemize}  
		%----
		\item[--] $P_0$ and $P_2$ send $\PadB{\wy}$ and $\Hash(\PadB{\wy}')$ respectively to $P_1$. 
		%---    
		\item[--] $P_0$ and $P_1$ send $\PadA{\wy}$ and $\Hash(\PadA{\wy}')$ respectively to $P_2$. 
		%---
		\item[--] $P_1$ and $P_2$ send $\Mask{\wy}'$ and $\Hash(\Mask{\wy}'')$ respectively to $P_0$. 
		%---
	\end{myitemize}
	%----  
	$P_i$ for $i \in \PInSet$ $\abort$ if the received values mismatch. Else $P_i$ sets $\wy = \Mask{\wy} - \PadA{\wy} - \PadB{\wy}$.	
\end{myprotocolbox}
%-----
Now the input sharing and output reconstruction stages in $\PiMal$ are similar to those in $\PiSemi$ apart from protocols $\PiShS$ and $\PiRecS$ being replaced with $\PiShM$ and $\PiRecM$ respectively.  

%----------------------------------------------------------------
\paragraph{Circuit Evaluation Stage}
%---------------------------------------------------------------- 
Protocol $\PiAdd$ remains secure in the malicious setting as well since it involves local operations only. The challenge lies in turning  the multiplication protocol $\PiMultS$  to  one that tolerates malicious behaviour.  We start with the observation  that $\PiMultS$ suffers in two {\em mutually-exclusive}  ways in the face of one malicious corruption, each under different corruption scenario. When $P_0$ is corrupt, the only possible  violation in  $\PiMultS$ comes in the form of  sharing $\Gammaxy \neq \Pad{\wx} \Pad{\wy}$ during the offline phase.  When $P_1$ (or $P_2$) is corrupt, the violation occurs when  a wrong share of $\Mask{\wz}$ is handed over  to the fellow honest evaluator during the online phase, causing  reconstruction of a wrong $\Mask{\wz}$. 
While the attacks are quite distinct in nature following the asymmetric roles played by the two sets  $\{P_0\}$ and $\{P_1,P_2\}$ in  $\PiMultS$, our novel construction solves both issues at the same time via checking product-relation of a single $\shrd$-shared  triple. We start with the technique to tackle a corrupt evaluator ($P_1$ or $P_2$) during the online phase. To identify if an incorrect $\Mask{\wz}$ is reconstructed by an honest evaluator, say $P_1$, he can seek the help of $P_0$ as follows: $P_1$ can send $\Mask{\wx}, \Mask{\wy}$ to $P_0$, who can then compute $\Mask{\wz}$, as $P_0$ already has knowledge of $\Pad{\wx}, \Pad{\wy}$ and $\Pad{\wz}$ from the offline phase and send back to $P_1$. Note that sending $\Mask{\wx}, \Mask{\wy}$ in clear to $P_0$ breaks privacy of the scheme and hence $P_1$ sends padded version of the same to $P_0$, namely $\Mstar{\wx} = \Mask{\wx} + \MPad{\wx}$ and $\Mstar{\wy} = \Mask{\wy} + \MPad{\wy}$. $P_0$ then computes $\Mstar{\wz} = -  \Mstar{\wx} \Pad{\wy} - \Mstar{\wy} \Pad{\wx} + \Pad{\wz} + 2\Gammaxy$. Note that,
%----
\begin{align*}
	\Mbar{\wz} &= - \Mstar{\wx} \Pad{\wy} - \Mstar{\wy} \Pad{\wx} + \Pad{\wz} + 2\Gammaxy\\
	&= -  (\Mask{\wx} + \MPad{\wx}) \Pad{\wy} - (\Mask{\wy} + \MPad{\wy}) \Pad{\wx} + \Pad{\wz} + 2\Gammaxy\\
	&= (\Mask{\wz} - \Mask{\wx} \Mask{\wy}) - (\MPad{\wx} \Pad{\wy} + \MPad{\wy} \Pad{\wx} - \Gammaxy)\\
	&= (\Mask{\wz} - \Mask{\wx} \Mask{\wy}) - \chi
\end{align*}
%----
Assuming that $P_0$ knows $\chi = \MPad{\wx} \Pad{\wy} + \MPad{\wy} \Pad{\wx} - \Gammaxy$, he can then compute $\Mbar{\wz} + \chi$ and send it back to $P_1$. Given the knowledge of $\Mask{\wx}, \Mask{\wy}$, $P_1$ can verify the correctness of $\Mask{\wz}$. The case for a honest $P_2$ follows similarly. Now we describe how to enable $P_0$ obtain $\chi = \MPad{\wx} \Pad{\wy} + \MPad{\wy} \Pad{\wx} - \Gammaxy$. First of all, note that revealing $\chi$ in clear to $P_0$ leads to breach of privacy. Because, $P_0$ knows $\Pad{\wx} , \Pad{\wy}, \Gammaxy$ from the offline phase and he receives $\Mask{\wx} + \MPad{\wx} , \Mask{\wy} + \MPad{\wy}$ during the online phase. With this information, $P_0$ can deduce a relation between $\Mask{\wx}$ and $\Mask{\wy}$. Hence, we tweak the value of $\chi$ to $\MPad{\wx} \Pad{\wy} + \MPad{\wy} \Pad{\wx} + \MPad{\wz} - \Gammaxy$ incorporating a random mask  $\MPad{\wz}$. To generate $\chi$, in the offline phase, parties $P_1, P_2$ locally sample random elements $\MPad{\wx},\MPad{\wy}, \MPad{\wz} \in \Z{\ell}$, compute a $\sqd$-sharing of $\chi$ and sends the shares to $P_0$. Let $\sqr{\chi}_{P_i} = \chi_i$ for $i \in \EInSet$. $P_0$ locally adds the $\sqd$-shares and obtains $\chi$. In the above step, a corrupt evaluator can introduce an error while computing the $\sqd$-share of $\chi$, affecting the correctness of the protocol. Thus, it is crucial to ensure the correctness of $\chi$ computed by $P_0$. 

To summarize, we now have two issues to tackle in the offline phase--  (i) as we pointed out earlier, during the offline phase, a corrupt $P_0$ can incorrectly share $\Gammaxy$; (ii) a corrupt evaluator can send a wrong $\sqd$-share of $\chi$ to $P_0$. Towards tackling these, once $P_0$ obtains the value $\chi$, parties locally compute $\shrd$-shares of values $\wma = \MPad{\wx} - \Pad{\wx}, \wmb = \MPad{\wy} - \Pad{\wy}$ and $\wmc = (\MPad{\wz} + \MPad{\wx} \MPad{\wy}) - \chi$ as follows: 
%-----
\vspace{-2mm}
\begin{center}
	\resizebox{0.49\textwidth}{!}
	{
		\begin{tabu} to 1\textwidth {  l  l  l }
			%-----
			$\shr{\wma}_{P_0} = (\PadA{\wx}, \PadB{\wx})$,
			& $\shr{\wmb}_{P_0} = (\PadA{\wy}, \PadB{\wy})$, 
			& $\shr{\wmc}_{P_0} = (\chi_1, \chi_2)$\\ 
			%-----
			$\shr{\wma}_{P_1} = (\MPad{\wx}, \PadA{\wx})$, 
			& $\shr{\wmb}_{P_1} = (\MPad{\wy}, \PadA{\wy})$,
			& $\shr{\wmc}_{P_1} = (\MPad{\wz} + \MPad{\wx} \MPad{\wy}, \chi_1)$\\
			%-----
			$\shr{\wma}_{P_2} = (\MPad{\wx}, \PadB{\wx})$
			& $\shr{\wmb}_{P_2} = (\MPad{\wy}, \PadB{\wy})$
			& $\shr{\wmc}_{P_2} = (\MPad{\wz} + \MPad{\wx} \MPad{\wy}, \chi_2)$\\
			%-----
		\end{tabu}
	}
\end{center}
%-----
\vspace{2mm}
Now   $(\shr{\wma}, \shr{\wmb}, \shr{\wmc})$ is a multiplication triple  ($\wmc = \wma \wmb$) if and only if  $P_0$ shares $\Gammaxy$ correctly (when it is corrupt) and $P_0$ reconstructs $\chi$ correctly (when one of the evaluators is corrupt). This is because, 
%----
\begin{align*}
	\wma \wmb &= (\MPad{\wx} - \Pad{\wx})(\MPad{\wy} - \Pad{\wy}) = \MPad{\wx} \MPad{\wy} + \Pad{\wx} \Pad{\wy} - \MPad{\wx} \Pad{\wy} - \MPad{\wy} \Pad{\wx}\\
	&= (\MPad{\wx} \MPad{\wy} + \MPad{\wz}) - (\MPad{\wx} \Pad{\wy} + \MPad{\wy} \Pad{\wx} + \MPad{\wz} - \Gammaxy)\\
	&= (\MPad{\wx} \MPad{\wy} + \MPad{\wz}) - \chi = \wmc
\end{align*}

%----
We first recall the two standard components needed to check the validity of a multiplication triple-- i) a tool for generating $\shrd$-shared random multiplication triple and ii) a technique to check  securely  the product relation of a $\shrd$-shared triple, given a valid $\shrd$-shared multiplication triple (often referred to as sacrificing technique). With a lot of constructions specifically available for the former one \cite{FLNW17,ABFLLNOWW17},  we choose to model it as an ideal functionality $\FTRIPLES$ and use it for our purpose without going into the details. For the latter component, we quickly recall the known protocol.

%----- 
\begin{myprotocolbox}{}{Protocol $\PiTripCheck$ to check product-relation of a triple}{fig:PiCheck}
	%----
	\justify
	\begin{myitemize}
		%--
		\item[--] Parties locally compute $\shr{\rho} = \shr{\wma} - \shr{\wmd}$ and $\shr{\sigma} = \shr{\wmb} - \shr{\wme}$.
		%---
		\item[--] Parties reconstruct $\rho$ and $\sigma$ by executing $\PiRecM(\shr{\rho}, \Partyset)$ and \newline $\PiRecM(\shr{\sigma}, \Partyset)$ respectively. 
		%---
		\item[--] Parties locally compute
		$\shr{\tau} = \shr{\wmc} - \shr{\wmf} - \sigma \shr{\wmd} - \rho \shr{\wme} - \sigma \rho$.
		%---
		\item[--] Parties reconstruct $\tau$ by executing $\PiRecM(\shr{\tau}, \Partyset)$ and output $\bot$, if $\tau \neq 0$.
		%---
	\end{myitemize}
	%----
\end{myprotocolbox}
%----   

$\FTRIPLES$, by now  a standard functionality \cite{FLNW17,ABFLLNOWW17},  allows to generate a set of $\shrd$-sharing of multiplication triples over $\Partyset$,  each of which, say $(\wmd, \wme, \wmf)$ satisfies the following-- i)  $\wmd, \wme$ and $\wmf$ are random and private and ii) $\wmf = \wmd \wme$. In Appendix~\ref{app:FTRIPLES}, we present an instantiation of this functionality, namely $\PiTripGen$ (\boxref{fig:TripGenII}), using the techniques proposed by \cite{FLNW17,ABFLLNOWW17}.

Protocol $\PiTripCheck$ \cite{CP17,FLNW17} (`prc' stands for product-relation check) takes a pair of $\shrd$-shared random and private triples as input, say $(\wma, \wmb, \wmc)$ and $(\wmd, \wme, \wmf)$, over $\Z{\ell}$, verifies if the former is a multiplication triple or not and nothing beyond, given the latter is a valid triple. The protocol appears in \boxref{fig:PiCheck} and its properties in Appendix~\ref{app:PiTripCheck}. 

By exploiting the definition of $\shrd$-sharing,  we reduce  the cost of $\PiTripCheck$ to just $2$, instead of $3$, instances of $\PiRecM$, in an amortized sense. Recall that the goal of the third invocation of $\PiRecM$ inside $\PiTripCheck$ is to reconstruct $\shr{\tau} = (\Mask{\tau}, \sqr{\Pad{\tau}})$, followed by checking if $\tau = 0$. It follows that $\tau = 0$ if and only if $\Mask{\tau} - \Pad{\tau} = 0$   implying $\Mask{\tau} = \Pad{\tau}$. Hence checking $\tau = 0$ is equivalent to checking if $\Mask{\tau} = \PadA{\tau} + \PadB{\tau}$, which can be translated to three  pair-wise checks -- (i) $P_0$ and $P_1$ can verify if $\Mask{\tau} - \PadA{\tau} \iseq \PadB{\tau}$; (ii) $P_1$ and $P_2$ can verify if $\Mask{\tau} - \PadB{\tau} \iseq \PadA{\tau}$; (iii) $P_0$ and $P_2$ can verify if $\Mask{\tau} - \PadB{\tau} \iseq \PadA{\tau}$. Parties in $\Partyset$ can mutually perform the above checks for all the instances of $\PiTripCheck$ together at the end by exchanging hash of all the required values.

%----
\begin{myprotocolbox}{}{Protocol $\PiMultM(\Wxyz)$:
	}{fig:PiMultM}
	
	%----          
	\justify
	\algoHead{Offline}:
	%-------
	\begin{myitemize}
		%----
		\item[--] Parties $P_0, P_1$ locally sample random $\PadA{\wz}, \GammaxyA \in \Z{\ell}$, while $P_0, P_2$ locally sample a random $\PadB{\wz}$. $P_0$ locally computes $\Gammaxy = \Pad{\wx} \Pad{\wy}$ and sends $\GammaxyB = \Gammaxy - \GammaxyA$ to $P_2$.
		%----
		\item[--] Parties execute $\PiTripGen$ to generate triple $(\shr{\wmd}, \shr{\wme}, \shr{\wmf})$.
		%----
		\item[--] Parties $P_1, P_2$ locally sample random $\MPad{\wx}, \MPad{\wy}, \MPad{\wz} \in \Z{\ell}$ and compute $\sqr{\MPad{\wz}}$ non-interactively.    
		%----
		\item[--] $P_i$ for   $i \in \EInSet$ computes $\sqr{\chi}_{P_i} = \MPad{\wx} \sqr{\Pad{\wy}}_{P_i} +  \MPad{\wy} \sqr{\Pad{\wx}}_{P_i} + \allowbreak \sqr{\MPad{\wz}}_{P_i} - \sqr{\Gammaxy}_{P_i}$ and sends $\sqr{\chi}_{P_i}$ to $P_0$, who computes $\chi$. 
		%----
		\item[--] Parties locally compute the $\shrd$-shares of the values $\wma = \MPad{\wx} - \allowbreak \Pad{\wx}, \wmb = \MPad{\wy} - \Pad{\wy}$ and $\wmc = (\MPad{\wz} + \MPad{\wx} \MPad{\wy}) - \chi$.
		%----
		\item[--] Parties execute $\PiTripCheck$ on $(\shr{\wma}, \shr{\wmb}, \shr{\wmc})$ and $(\shr{\wmd}, \shr{\wme}, \shr{\wmf})$.
		%----
	\end{myitemize}
	%------  
	
	\justify          
	\algoHead{Online}:  
	%------
	\begin{myitemize}  
		%----
		\item[--] $P_i$ for $i \in \EInSet$ locally computes $\sqr{\Mask{\wz}}_{P_i} = (i-1)\Mask{\wx} \Mask{\wy} - \allowbreak \Mask{\wx} \sqr{\Pad{\wy}}_{P_i} - \Mask{\wy} \sqr{\Pad{\wx}}_{P_i} + \sqr{\Pad{\wz}}_{P_i} +  \sqr{\Gammaxy}_{P_i}$. $\ESet$ mutually exchange their shares and reconstruct $\Mask{\wz}$.
		%----
		\item[--] $P_1$ sends $\Mstar{\wx} = \Mask{\wx} + \MPad{\wx}, \Mstar{\wy} = \Mask{\wy} + \MPad{\wy}$ to $P_0$, while $P_2$ sends $\Hash(\Mstar{\wx} || \Mstar{\wy})$ to $P_0$.  $P_0$ outputs $\bot$, if the received values are inconsistent.
		%---
		\item[--] $P_0$ computes $\Mstar{\wz} = -  \Mstar{\wx} \Pad{\wy} - \Mstar{\wy} \Pad{\wx} + \Pad{\wz} + 2\Gammaxy + \chi$ and sends $\Hash(\Mstar{\wz})$ to both $P_1$ and $P_2$.
		%---
		\item [--] $P_i$ for $i \in \EInSet$ $\abort$ if $\Hash(\Mstar{\wz}) \ne \Hash(\Mask{\wz} - \Mask{\wx} \Mask{\wy} + \MPad{\wz})$.
	\end{myitemize}
	%----  
\end{myprotocolbox}
%------------------

With the building blocks set, we present our maliciously-secure multiplication protocol $\PiMultM$ in \boxref{fig:PiMultM}. Note that the use of hash function improves the amortized cost in the online phase of $\PiMultM$-- (i) $P_2$ can send a single hash of all the $\Mstar{\wx}$ and $\Mstar{\wy}$ values for all the instances of $\PiMultM$ to $P_0$ in the end of the circuit-evaluation;  (ii)  $P_0$ can send a single hash of all the $\Mstar{\wz}$ values for all the instances of $\PiMultM$ to the evaluators at the end of the circuit-evaluation. The former step can be coupled with the communication of $(\Mstar{\wx}, \Mstar{\wy})$ by $P_1$ to $P_0$. Party $P_1$ sending to $P_0$ attributes to the increase of the communication cost per multiplication gate in the malicious setting, compared to the semi-honest setting. On the positive note, coupling the above communication for all the multiplication gates together results in a couple of rounds overhead compared to the semi-honest protocol.  As a consequence, the latency of the malicious protocol remains as good as the semi-honest protocol. 

The correctness of the protocol $\PiMultM$ is stated in Lemma \ref{lemma:PiMultM}.
%----
\begin{lemma}[Correctness]
	\label{lemma:PiMultM}
	%----
	In the protocol $\PiMultM$, the following holds:
	During the offline phase, if $P_0$ is corrupt and $\sqd$-shares $\Gammaxy \neq \Pad{\wx} \Pad{\wy}$, then the honest evaluators output $\bot$. On the other hand, if one of the evaluators is corrupt and enforces the honest $P_0$ to obtain an incorrect $\chi$, then the honest parties output $\bot$.  During the online step, if one of the evaluators is corrupt and enforces the honest evaluator to obtain an incorrect $\Mask{\wz}$, then the honest evaluator outputs $\bot$         
\end{lemma}
%-----
\begin{proof}
	%----   
	For correctness, first consider the case when $P_0$ is {\it corrupt} and $\sqd$-shares $\Gammaxy \neq \Pad{\wx} \Pad{\wy}$ during offline step. Let $\Gammaxy = \Pad{\wx} \Pad{\wy} + \Delta$ where $\Delta$ is the error introduced by $P_0$. Now, 
	%----
	%{\small
	\begin{align*}
		\wmc &= (\MPad{\wx} \MPad{\wy} + \MPad{\wz}) - (\MPad{\wx} \Pad{\wy} + \MPad{\wy} \Pad{\wx} + \MPad{\wz} - (\Gammaxy - \Delta))\\
		&= (\MPad{\wx} - \Pad{\wx})(\MPad{\wy} - \Pad{\wy}) - \Delta = \wma - \Delta \ne \wma
	\end{align*}
	%}
	%----
	and thus $(\wma, \wmb, \wmc)$ is not a multiplication triple. Then, from Lemma~\ref{lemma:TripCheck}, honest evaluators output $\bot$.
	
	Second, we consider the case when one of the evaluators, say $P_1$, sends $\chi_1 + \Delta$ to $P_0$ who reconstructs $\chi' = \chi + \Delta$. Then, the value 	
	%----
	%{\small
	\begin{align*}
		\wmc &=  (\MPad{\wx} \MPad{\wy} + \MPad{\wz}) - \chi' = (\MPad{\wx} \MPad{\wy} + \MPad{\wz}) - (\chi + \Delta)\\
		&= (\MPad{\wx} \MPad{\wy} + \MPad{\wz}) - (\MPad{\wx} \Pad{\wy} + \MPad{\wy} \Pad{\wx} + \MPad{\wz} - \Gammaxy) - \Delta\\
		&= (\MPad{\wx} - \Pad{\wx})(\MPad{\wy} - \Pad{\wy}) - \Delta = \wma - \Delta \ne \wma
	\end{align*}
	%}
	%----
	and hence $(\wma, \wmb, \wmc)$ is not a multiplication triple. Thus, similar to the previous case, honest parties output $\bot$.
	
	Lastly, we consider the case, when one of the evaluators, say $P_1$, is {\it corrupt} and during online step sends $\sqr{\Mask{\wz}}_{P_1} + \Delta$ for some non-zero $\Delta$ during the reconstruction, so that $P_2$ reconstructs $\Mask{\wz} + \Delta$, instead of $\Mask{\wz}$. In this case, the honest $P_0$ would have $\chi = \MPad{\wx} \Pad{\wy} + \MPad{\wy} \Pad{\wx} + \MPad{\wz} - \Gammaxy$  from offline step. Moreover, during online step, $P_0$ correctly learns $\Mstar{\wx} = \Mask{\wx} + \MPad{\wx}$ and $\Mstar{\wy} = \Mask{\wy} + \MPad{\wy}$. Furthermore, $\Gammaxy = \Pad{\wx}  \Pad{\wy}$ holds. It then follows that $\Mstar{\wz}$ received by $P_2$ from $P_0$ will be different from $\Mask{\wz} + \Delta - \Mask{\wx} \Mask{\wy} + \MPad{\wz}$ locally computed by $P_2$ and hence $P_2$ will output $\bot$.
	%----
\end{proof}
%----
The informal privacy argument of $\PiMultM$ is as follows. We first consider the case when $P_0$ is corrupt, where $\shr{\wx}, \shr{\wy}$ and  $\shr{\wz}$ are defined by the shares of $P_1, P_2$. The privacy for this case follows from the fact that $P_0$ does not learn anything about $\Mask{\wx}, \Mask{\wy}$ and $\Mask{\wz}$, neither during the offline step, nor during the online step.  Clearly, the communication between $P_0$ and $P_1, P_2$ during offline step is independent of $\Mask{\wx}, \Mask{\wy}$ and $\Mask{\wz}$. Moreover, the value $\chi$ reveals nothing about $\MPad{\wx}$ and $\MPad{\wx}$ since it is padded with a random $\MPad{\wz}$. During the online step, $P_0$ learns $\Mstar{\wx}$ and $\Mstar{\wy}$, which reveals nothing about $\Mask{\wx}, \Mask{\wy}$, as  $\MPad{\wx}$ and $\MPad{\wy}$ remains random and private for $P_0$. We next consider the case when one of the evaluators, say $P_1$ is corrupt. The privacy for this case follows from the fact that $\Pad{\wx}, \Pad{\wy}, \Pad{\wz}$ and $\Gammaxy$ remains private from the view point of $P_1$. On the other hand, no additional information is revealed from $\Mstar{\wz}$ during the online step, as adversary will already know that $\Mstar{\wz} = \Mask{\wz} - \Mask{\wx} \Mask{\wy} + \MPad{\wz}$.

We present  a detailed security proof for our 3PC protocol $\PiMal$  in Appendix~\ref{app:3PCMal}, showing that it emulates the functionality $\FTHREEMPCABORT$ as given in \boxref{fig:FTHREEMPCABORT}. 
%---------
\begin{mysystembox}{}{Functionality $\FTHREEMPCABORT$}{fig:FTHREEMPCABORT}
	\justify $\FTHREEMPCABORT$ interacts with the parties in $\Partyset$ and the adversary $\Sim$ and is parameterized by a $3$-ary function $f$, represented by a publicly known arithmetic circuit $\ckt$ over $\Z{\ell}$.
	%---
	\begin{description}
		%---
		\item {\bf Input: } Upon receiving the input $\wx_, \ldots, \wx_{\IS}$ from the respective parties in $\Partyset$, do the following: if $(\INPUT, *)$ message was received from $P_j$ corresponding to $\wx_j$, then ignore. Otherwise record $\wx_j' = \wx_j$ internally. If $\wx_j' \ne \Z{\ell}$, consider $\wx_j' = \abort$.
		%----
		\item {\bf Output to adversary: } If there exists $j \in \{1, \ldots, \IS \}$ such that $\wx_j' = \abort$, send $(\OUTPUT, \bot)$ to all the parties. Else, send $(\OUTPUT, (\wy_1,\ldots,\newline \wy_{\OS}))$ to the adversary $\Sim$, where $(\wy_1,\ldots,\wy_{\OS}) = f(\wx_1',\ldots,\wx_{\IS}')$.
		%--- 
		\item {\bf Output to selected honest parties: } Receive $(\SELECT, \{I\})$ from adversary $\Sim$, where $\{I\}$ denotes a subset of the honest parties. If an honest party belongs to $I$, send $(\OUTPUT, \bot)$, else send $(\OUTPUT, (\wy_1,\ldots,\wy_{\OS}))$.
		%---
	\end{description}
	%----
\end{mysystembox}
%---------

We now prove the communication complexity of protocol $\PiMal$ below.
%--------       
\begin{theorem}
	\label{thm:PiMal}
	Protocol $\PiMal$ has the following complexities.
	\begin{myitemize}
		%------
		\item[]{\bf Input-sharing Stage:} It is non-interactive during the offline phase and requires one round and an amortized communication of at most $2\IS$ ring elements during the online phase.
		%------
		\item[]{\bf Circuit-evaluation Stage:} Assuming $\MS = 2^{20}$ and a statistical security parameter $s = 40$,  in the amortized sense, evaluating each multiplication gate requires $4$ rounds and communication of $21$ ring elements  in the offline phase, while the online phase needs $1$ round with a communication of $4$ ring elements. 
		%------
		\item[]{\bf Output-reconstruction Stage:} It requires one round and an amortized communication of $3\OS$  ring elements.
		%------
	\end{myitemize}
\end{theorem}
%--------------
\begin{proof}
	The complexity for the Input-sharing Stage follows from  Theorem \ref{thm:PiSemiCost} and the fact that the cost of $\PiShM$ reduces to that of $\PiShS$ in an amortized sense due to the use of the hash function. During the circuit-evaluation stage, the addition gates need no interaction, as usual.  For a multiplication gate, the offline communication include-- (i) sending a share of   $\sqr{\Gammaxy}$ to $P_2$; (ii) the amortized cost of generating one shared triple via $\FTRIPLES$; (iii)  the cost of reconstructing $\chi$ towards $P_0$ and lastly (iv) the cost of one $\PiTripCheck$. The first one requires one round and communication of one element. The second one requires  $3$ rounds and an amortized communication of $9B-6$ ring elements, where $B = \frac{\sparam}{\log_2{\MS}}$, using the techniques of \cite{ABFLLNOWW17} (see Appendix \ref{app:FTRIPLES}), where $s$ is the statistical parameter dictating the performance of underlying cut-and-choose technique. Assuming $\MS = 2^{20}$, $s = 40$, this ensures that generating a single multiplication triple require $3$ rounds and an amortized communication of $12$ ring elements. The third one requires one round and communication of two elements. The fourth and last one requires one round and an amortized communication of $6$ elements as part of the two underlying instances of $\PiRecM$. This sums up to communication of $21$ elements per multiplication gate. 
	
	The total number of rounds for evaluating the multiplication gates during the offline phase turns to be $4$ as follows: $P_0$ can send the share of $\sqr{\Gammaxy}$ to $P_2$ and in parallel, the parties can start generating a shared triple via $\FTRIPLES$; while the former requires one round, the latter requires three rounds. Once the share of $\sqr{\Gammaxy}$ is available with $P_2$, party $P_1$ and $P_2$ can reconstruct $\chi$ towards $P_0$, requiring one round, which overlaps with the second round of the instantiation of $\FTRIPLES$. Once the third round of the instantiation of $\FTRIPLES$ is over, the parties execute the instance of $\PiTripCheck$, which requires one additional round.
	
	During the online phase, evaluating a multiplication gate requires one round and communication of two elements for the reconstruction of  $\Mask{\wz}$. Also, $P_1$ needs to send $\Mstar{\wx}$  and $\Mstar{\wy}$ values to $P_0$ per instance, which requires just one round for all the multiplication gates and communication of $2$ ring elements per gate. Summing up, evaluating a multiplication gate in the online phase requires  an amortized round complexity of $1$ and communication of $4$ elements.
	
	The output-reconstruction phase requires one round and an amortized communication of $3\OS$ elements, as  the cost of $\PiRecM$ reduces to $\PiRecS$ in an amortized sense due to the use of the hash function.
	%-------
\end{proof}

%----------------------------------------------------------------
\subsection{Achieving Fairness} \label{sec:fairness}
%----------------------------------------------------------------
We boost the security of $\PiMal$ from abort to fairness  via a fair reconstruction protocol $\PiFRec$ that substitutes $\PiRecM$ for the reconstruction of the circuit outputs. To fairly reconstruct $\shr{\wy}$,  the pair $\{P_0,P_1\}$ commit their common share  $\PadA{\wy}$ to $P_2$ and likewise  the pair $P_0,P_2$ commit their common share  $\PadB{\wy}$ to $P_1$ in the offline phase. In the online phase, the evaluator pair $\{P_1,P_2\}$  commit their common information $\Mask{\wy}$ to $P_0$. In all the three cases, shared random (PRF) key is used to derive  the randomness for preparing the commitments. As a result, each pair should prepare an identical commitment ideally. The recipient in each case can abort when the received commitments do not match. If no abort happens, $P_0$  signals $P_1$ and $P_2$ to start opening  the commitments which will help the parties to get their missing share and reconstruct the output.  As there is at least one honest party in each  pair of $(P_0, P_1), (P_0, P_2)$ and $(P_1, P_2)$, the opened value of the honest party from each pair is used for reconstructing $\wy$. Lastly, if the protocol aborts before, then none receive the output maintaining fairness.

A very subtle issue arises in the above protocol in the absence of broadcast channel. A corrupt  $P_0$ can send distinct signals to $P_1$ and $P_2$ (abort to one and continue to the other), breaching unanimity in the end. To settle this, we make   the pair $\{P_0,P_1\}$ to commit a value  $r_1$ chosen from their common random source to $P_2$ and likewise  the pair $P_0,P_2$ to  commit a common value  $r_2$ to $P_1$ in the offline phase. In the online phase, when $P_0$ signals abort to $P_1$, it sends the opening of $r_2$ along. Similarly, when $P_0$ signals abort to $P_2$, it sends the opening of $r_1$ along.  Now an evaluator, say $P_1$ on receiving the abort can convince $P_2$ that it has indeed received abort from $P_0$, using $r_2$ as the {\em proof of origin} for the abort message.  Because the only way $P_1$ can secure $r_2$ is via $P_0$. Put differently, a corrupt $P_1$ cannot simply claim that it received abort from $P_0$, while $P_0$ is really instructed to continue.  A single pair of $(r_1,r_2)$ can be used as a proof of origin for multiple instances of reconstruction running in parallel. Protocol $\PiFRec(\shr{\wy}, \Partyset)$ is formally presented in \boxref{fig:fRec}.
%-----
\begin{myprotocolbox}{}{Protocol $\PiFRec(\shr{\wy}, \Partyset)$}{fig:fRec}
	%-----
	%--------
	\justify
	\algoHead{Offline:}  
	%-----               
	\begin{myitemize}
		%----
		\item[--] Parties $P_0, P_1$ locally sample a random $r_1 \in \Z{\ell}$, prepare and send  commitments of $\PadA{\wy}$ and $r_1$ to $P_2$. Similarly, parties $P_0, P_2$ sample $r_2$ and send  commitments of $\PadB{\wy}$ and $r_2$ to $P_1$ The randomness needed for both commitments are sampled from the PRF key-setup.
		%----
		\item[--] $P_1$ (resp.  $P_2$) aborts if the received commitments mismatch.
		%---
	\end{myitemize}
	%--------
	\justify
	\algoHead{Online:} 
	%-----               
	\begin{myitemize}
		%----
		\item[--] $P_1, P_2$ compute a commitment of $\Mask{\wy}$ using randomness sampled from their PRF key-setup and send it to $P_0$.
		%----
		\item[--] If the commitments do not match, $P_0$ sends $(\abort, o_1)$ to $P_2$, while he sends $(\abort, o_2)$  to $P_1$ and aborts, where $o_i$ denotes opening information for the commitment of $r_i$. Else $P_0$ sends $\continue$  to both $P_1$ and $P_2$.
		%----
		\item[--] $P_1, P_2$ exchange the messages received from $P_0$. 
		%----
		\item[--] $P_1$ aborts if he receives either (i) $(\abort, o_2)$ from $P_0$ and $o_2$ opens the commitment of $r_2$ or (ii) $(\abort, o_1)$ from $P_2$ and $o_1$ is the correct opening information of $r_1$. The case for $P_2$ is similar to that of $P_1$
		%----
		\item[--] If no abort happens, parties obtain their missing share of $a$ as follows:
		%----
		\begin{myitemize}
			\item[--] $P_0, P_1$ open $\PadA{\wy}$ towards $P_2$.
			%--- 
			\item[--] $P_0, P_2$ open $\PadB{\wy}$ towards $P_1$.
			%---
			\item[--] $P_1, P_2$ open $\Mask{\wy}$ towards $P_0$.
			%---
		\end{myitemize}
		%----
		\item[--] Parties reconstruct the value $\wy$ using missing share that matches with the agreed upon commitment.
	\end{myitemize}
	%----
\end{myprotocolbox}
%---

The complexity of $\PiFRec$  is stated below. The commitment can be implemented via a hash function $\mathcal{H}()$ e.g.  $(c, o) =  (\mathcal{H}(x||r), \allowbreak x||r) = Com \allowbreak(x; r)$, whose  security can be proved in the random-oracle model (ROM) \cite{KL14}. We do not include the cost of commitment and opening of $r_1$ and $r_2$, as they will get amortized away over many instances of $\PiFRec$.  
\begin{lemma}
	Protocol $\PiFRec$ requires one round and an amortized  communication of $4$ commitments in the offline phase. $\PiFRec$ requires four rounds and an amortized  communication of at most $2$ commitments and $6$ opening of commitments  in the online phase. 
\end{lemma}
%----
%\clearpage
\section{Privacy Preserving Machine Learning}
\label{sec:privML}
We apply our techniques for 3PC developed so far to the regime of ML prediction for a range of prediction functions-- linear regression, logistic regression, linear SVM classification, and linear SVM regression. 
%----------------------------------------------------------------
\subsection{The Model}
\label{sec:Model}
%----------------------------------------------------------------
A model-owner $\Model$,  holding a vector of {\em trained model parameters}, would like to offer ML prediction service to  a client $\Client$ holding a {\it query vector} as per certain prediction function. In the server-aided setting, $\Model$ and $\Client$ outsource their respective inputs in shared fashion  to three untrusted but non-colluding servers $\{P_0, P_1, P_2\}$ who perform the computation in shared fashion via techniques developed for our 3PC protocols and reconstruct the output to the client alone. The client learns the output and nothing beyond.  We assume a {\it computationally bounded} adversary $\Adv$, who  can corrupt at most one of the servers $\{P_0, P_1, P_2\}$ and one of  $\{\Model, \Client\}$ in either semi-honest or malicious fashion.  The security against an $\Adv$ corrupting parties in both sets $\{P_0, P_1, P_2\}$ and $\{\Model, \Client\}$ semi-honestly and likewise maliciously reduces to the semi-honest and respectively malicious security of our 3PC protocols. Adversarial machine learning \cite{TramerZJRR16,Papernot17, KNOCKOFF18} that includes attacks launched by a client to learn the model using its outputs, lies outside the scope of this work.   
Following the existing literature on server-aided secure ML \cite{KamaraMR11, NikolaenkoWIJBT13, NikolaenkoIWJTB13, GasconSB0DZE16}, we do not count the cost of $\Model$ and $\Client$ making their inputs available in secret-shared form amongst the servers and the cost of reconstructing the output to the client. We assume that the inputs are available to the servers in a secret-shared form and focus  on efficient computation of a prediction function on the shared inputs to obtain shared outputs.
%----------------------------------------------------------------
\subsection{Notations}
\label{sec:prelimsML}
%----------------------------------------------------------------
For a vector $\MA$, $\Va_i$ denotes the $i^{th}$ element in the vector.  For two vectors $\MA$ and $\MB$ of length $d$, their scalar dot product is $\MA \band \MB = \sum_{i = 1}^{d} \Va_i \Vb_i$. The definitions of $\sqd$-sharing and $\shrd$-sharing are extended in a natural way for the vectors. A vector $\MA = (\Va_1, \ldots, \Va_{d})$ is said to be $\sqd$-shared, denoted as $\sqr{\MA}$, if each $\Va_i$ is  $\sqd$-shared. We use the notations $\sqr{\MA}_{P_1}$ and $\sqr{\MA}_{P_2}$ to denote the vector of $\sqd$-shares of $P_1$ and $P_2$ respectively, corresponding to $\sqr{\MA}$. Similarly, a vector $\MA = (\Va_1, \ldots, \Va_{d})$ is said to be $\shrd$-shared, denoted as $\shr{\MA}$, if each $\Va_i$ is $\shrd$-shared. We use the notation $\VPad{\Va}$ and $\VMask{\Va}$ to denote the vector of masks and vector of masked values corresponding to $\shr{\MA}$. Finally, we note that the linearity of $\sqd$ and $\shrd$-sharings hold even over vectors.
%----------------------------------------------------------------
\subsection{Fixed Point Arithmetic}
\label{sec:FPA}
%----------------------------------------------------------------
We represent decimal values as $\ell$-bit integers in signed $2$'s complement representation with the most significant bit representing the sign bit and $x$ least significant bits  representing the fractional part. For our purpose, we choose $\ell = 64$ and $x = 13$, keeping $i = \ell - x - 1 = 50$ bits for the integral part of the value. We then treat these  $\ell$-bit strings as elements of $\Z{\ell}$. A product of two numbers from this domain would lead to expanding $x$ to $26$ and yet leaving $37$ bits for the integer part which keeps the accuracy unaffected.  As the prediction functions of our concern require multiplication of depth one, the prediction function output values have the above format. Noticeably, since SecureML \cite{MohasselZ17} and ABY3 \cite{MR18} need to do multiplication in sequence multiple times for the task of training, they propose a new  method of truncation to maintain a  representation invariant across the sequential products. This is necessary to keep accuracy in check in their works. 

%----------------------------------------------------------------
\subsection{Protocols for ML}
\label{sec:ProtML}
%----------------------------------------------------------------
We begin with some of the building blocks required.
%----------------------------------------------------------------
\paragraph{Secure Dot Product}
%----------------------------------------------------------------
Given the $\shrd$-shares of $d$ element vectors $\MP$ and $\MQ$, the goal of a secure dot-product  is to compute  $\shrd$-sharing of $\MP \band \MQ$. Using $\PiMult$ naively to compute the product of each component would require a communication complexity that is linearly dependent on $d$ in both the offline and online phase. In the {\it semi-honest} setting, following the literature  \cite{CDI05,CatrinaH10,HooghSCA14,MR18,RiaziWTS0K18}, we make the communication of $\PiDot$ independent of $d$ as follows: during the offline phase, $P_0$ $\sqd$-shares  {\it only} $\gamma_{\VPQ} = \VPad{p} \band \VPad{q}$, instead of each individual  $\Pad{\Vp_i} \Pad{\Vq_i}$. During the online phase, instead of reconstructing each $\Mask{\Vp_i \Vq_i}$ separately to  compute $\Mask{\Vr}$ where $\Vr =  \MP \band \MQ$, the evaluators $P_1, P_2$ locally compute $\sqr{\Mask{\Vr}}$ and then reconstruct $\Mask{\Vr}$. We call the resultant protocol as $\PiDotSemi$ (\boxref{fig:PiDotSemi}).
%----
\begin{myprotocolbox}{}{Protocol $\PiDotSemi$}{fig:PiDotSemi}
	%-----
	\justify 
	\algoHead{Offline}: $P_0, P_1$ sample random $\PadA{\Vr}, \GammaA{\VPQ} \in \Z{\ell}$, while $P_0, P_2$ sample random $\PadB{\Vr} \in \Z{\ell}$. $P_0$ locally computes $\GammaV{\VPQ} = \VPad{p} \band \VPad{q}$, sets $\GammaB{\VPQ} = \allowbreak \GammaV{\VPQ} - \GammaA{\VPQ}$ and sends $\GammaB{\VPQ}$ to $P_2$.
	
	%------    
	\justify 
	\algoHead{Online}: 
	\begin{myitemize}
		%---
		\item[--] $P_i$ for $i \in \EInSet$ locally computes $\sqr{\Mask{\Vr}}_{P_i} = \sum_{j=1}^{d} \big( (i-1) \Mask{\Vp_j} \Mask{\Vq_j} - \Mask{\Vp_j} \sqr{\Pad{\Vp_j}}_{P_i} - \Mask{\Vq_j} \sqr{\Pad{\Vp_j}}_{P_i} \big) + \sqr{\GammaV{\VPQ}}_{P_i} + \sqr{\Pad{\Vr}}_{P_i}$.
		%---
		\item[--] $P_1$ and $P_2$ mutually exchange $\sqr{\Mask{\Vr}}$ to reconstruct $\Mask{\Vr}$.
		%---
	\end{myitemize}
	%-----	
\end{myprotocolbox}
%----

Due to the extra checks we introduce for tolerating a maliciously adversary in our multiplication protocol, the optimization done above for semi-honest protocol in the offline phase does not work. As a result, we resort to $d$ invocations of our multiplication protocol. Invoking Theorem~\ref{thm:PiMal}, our protocol for dot product then needs to communicate $21d$ ring elements in the offline phase. However,  we improve the online cost from $4d$ (as per Theorem~\ref{thm:PiMal}) to $2d + 2$ as follows. The parties execute the online stage of protocol $\PiDotSemi$. In parallel, $P_1$ sends $\Mstar{\Vp_i}, \Mstar{\Vq_i}$ for $i \in \{1, \ldots, d\}$ to $P_0$, while $P_2$ sends the corresponding hash to $P_0$. Instead of sending $\Mstar{\Vp_i \Vq_i}$ for each  $\Vp_i \Vq_i$, $P_0$ can ``combine" all the  $\Mstar{\Vp_i \Vq_i}$ values and send a single $\Mstar{\Vr}$ to $P_1, P_2$ for verification. In detail, $P_0$ can compute $\Mstar{\Vu} = \sum_{j=1}^{d} \Mstar{\Vu_j}$ and send a hash of the same to both $P_1$ and $P_2$, who can then cross check with a hash of $\Mask{\Vr} - \sum_{j=1}^{d} (\Mask{\Vp_j} \Mask{\Vq_j} - \MPad{\Vr_j})$.  We call the resultant protocol as $\PiDotMal$ and the communication complexity is given below.
\begin{lemma}
	\label{lemma:PiDotCC}
	%-----
	$\PiDotSemi$ requires communication of one ring element during the offline step and communication of two ring elements in online step. $\PiDotMal$ requires communication of $21 d$ ring elements during the offline step and communication of $2 d + 2$ ring elements in online step.
	%-----
\end{lemma}
%------------

%----------------------------------------------------------------
\paragraph{Secure Comparison}
%----------------------------------------------------------------
Comparing two arithmetic values is one of the major hurdles in realizing efficient secure ML algorithms. Given arithmetic shares $\shr{\su}, \shr{\sv}$, parties wish to check whether $\su < \sv$, which is equivalent to checking if $\sval < 0$, where $\sval = \su - \sv$. In the fixed-point arithmetic representation, this task can be accomplished by checking  the $\MSB{\sval}$. Thus the goal reduces to generating boolean-shares of $\MSB{\sval}$ given the arithmetic-sharing $\shr{\sval}$. Here, we exploit the asymmetry in our secret sharing scheme and forgo expensive primitives such as garbled circuits or parallel prefix adders, which are used in SecureML \cite{MohasselZ17} and ABY3 \cite{MR18}.

%----------------
\begin{myprotocolbox}{}{Protocol $\PiBitExtS{\share{\sval}}{\Partyset}$}{fig:BitExtS}
	\justify 
	\algoHead{Offline:} $P_1, P_2$ together sample random $\sr, \sr' \in \Z{\ell}$ and set $\msp = \allowbreak \MSB{\sr}$. Parties non-interactively generate boolean shares of $\msp$ as $\shrB{\msp}_{P_0} = (0,0),\newline \shr{\msp}_{P_1} = (\msp, 0)$ and $\shr{\msp}_{P_2} = (\msp,0)$.
	%------
	\justify 
	\algoHead{Online:} $P_1$ sets $\sqr{\sval}_{P_1} = \Mask{\sval} - \PadA{\sval}$, $P_2$ sets $\sqr{\sval}_{P_2} = - \PadB{\sval}$. 
	%-----------
	\begin{myitemize}
		%---
		\item[--] $P_1$ sends $\sqr{\sr \sval}_{P_1} = \sr \sqr{\sval}_{P_1} + \sr'$ to $P_0$, while $P_2$ sends $\sqr{\sr \sval}_{P_2} = \allowbreak \sr \sqr{\sval}_{P_2} - \sr'$ to $P_0$, who adds them to obtain $\sr \sval$.
		%----
		\item[--] $P_0$ executes $\PiShS(P_0, \msq)$ to generate $\shr{\msq}$ where $\msq = \MSB{\sr \sval}$.
		%----
		\item[--] Parties locally compute $\shrB{\MSB{\sval}} = \shrB{\msp} \Xor \shrB{\msq}$.
		%----
	\end{myitemize}
	%----
\end{myprotocolbox}
%----------------

We observe that in the signed 2's complement representation, if we multiply two values, then the sign of the result is the sign of the underlying product. Consequently, if a value $\sval$ is multiplied with $\sr$, then $\sign(\sval \cdot \sr) = \sign(\sval) \xor \sign(\sr)$. On a high level, the semi-honest protocol (\boxref{fig:BitExtS}) proceeds as follows: $P_1,P_2$ reconstruct $\sr \sval$ towards $P_0$ where $\sval$ is the value we need the sign of, and $\sr$ is a random value sampled by $P_1,P_2$ together. $P_0$ in turn boolean-shares the sign of $\sr \sval$. Parties retrieve the sign of $\sval$ by XORing the sign of $\sr \sval$ with the sign of $\sr$. For the sake of clarity, we use the superscript {\bf B} to denote the boolean shares. 

%----------------
%\vspace{2mm}
\begin{myprotocolbox}{}{Protocol $\PiBitExtM{\share{\sval}}{\Partyset}$}{fig:BitExtM}
	\justify 
	\algoHead{Offline:} $P_1, P_2$ sample random $\sr_1 \in \Z{\ell}$ and set $\msp_1 = \MSB{\sr_1}$ while $P_0, P_2$ sample random $\sr_2$ and set $\msp_2 = \MSB{\sr_2}$. 
	\begin{myitemize}
		\item[--] Parties non-interactively generate $\shrd$-shares of $\sr_1$ as $\shr{\sr_1}_{P_0} = (0,0),\newline \shr{\sr_1}_{P_1} = (\sr_1, 0)$ and $\shr{\sr_1}_{P_2} = (\sr_1,0)$.
		%-----
		\item[--] Parties non-interactively generate $\shrd$-shares of $\sr_2$ as $\shr{\sr_2}_{P_0} = (0, -\sr_2),\newline \shr{\sr_2}_{P_1} = (0, 0)$ and $\shr{\sr_2}_{P_2} = (0, -\sr_2)$.
		%-----
		\item[--] Parties execute $\PiMultM$ on $\sr_1$ and $\sr_2$ to generate $\shr{\sr} = \shr{\sr_1 \sr_2}$.%, where $\sr = \sr_1 \sr_2$. 
		%-----
		\item[--] Parties non-interactively generate boolean shares of $\msp_1$ as $\shrB{\msp_1}_{P_0} = (0,0), \shrB{\msp_1}_{P_1} = (\msp_1, 0)$ and $\shrB{\msp_1}_{P_2} = (\msp_1,0)$.
		%-----
		\item[--] Parties non-interactively generate boolean shares of $\msp_2$ as $\shrB{\msp_2}_{P_0} = (0,\msp_2), \shrB{\msp_2}_{P_1} = (0, 0)$ and $\shrB{\msp_2}_{P_2} = (0, \msp_2)$.
		%-----
		\item[--] Parties locally compute $\shrB{\msp} = \shrB{\msp_1} \xor \shrB{\msp_2}$
		%-----
	\end{myitemize}
	%------   
	\justify 
	\algoHead{Online:}~
	%-----------
	\begin{myitemize}
		%---
		\item[--] Parties execute $\PiMultM$ on $\sr$ and $\sval$ to generate $\sr \sval$ followed by executing $\PiRecM(P_0, \sr \sval)$ and $\PiRecM(P_1, \sr \sval)$ to enable $P_0, P_1$ obtain $\sr \sval$.
		%----
		\item[--] $P_1$ execute $\PiShM(P_1, \msq)$ to generate $\shrB{\msq}$ where $\msq = \MSB{\sr \sval}$. In parallel, $P_0$ locally computes $\Mask{\msq}$ and sends $\Hash(\Mask{\msq})$ to $P_2$, who $\abort$ if the value mismatch with one received from $P_1$. 
		%----
		\item[--] Parties locally compute $\shrB{\MSB{\sval}} = \shrB{\msp} \Xor \shrB{\msq}$.
		%----
	\end{myitemize}
	%----
\end{myprotocolbox}
%----

For the malicious case, we cannot solely rely on $P_0$ to generate $\shrB{\MSB{\sr \sval}}$. The modified protocol for the malicious setting appears in \boxref{fig:BitExtM}. The correctness for the malicious version appears in Appendix \ref{app:PrivML}. The communication  and round complexity are given below.
\begin{lemma}
	\label{lemma:PiBitExt}
	%-----
	$\piBitExtS$ requires no communication during the offline step, while it requires two rounds and communication of $2 \ell + 2$ bits during the online step. $\piBitExtM$ requires four rounds and an amortized communication of $46 \ell$ bits during the offline step, while it requires three rounds and an amortized communication of $6 \ell + 1$ bits during the online step.
	%-----
\end{lemma}
%------------

%----------------------------------------------------------------
\subsection{ML Prediction Functions and Abstractions}
\label{sec:MLAlgo}
%----------------------------------------------------------------
We consider four prediction functions -- two from regression category with real or continuous value as the output and two from classification type with a bit as the output. The inputs to the functions are vectors of decimal values. We provide a high-level overview of the functions below and more details can be found in \cite{MohasselZ17, MakriRSV17, MR18}.
%------
\begin{myitemize}
	%----
	\item[$\circ$] {\bf Linear Regression}: Model $\Model$ owns a $d$-dimensional model parameter $\MW$ and a bias $\Vb$, while client $\Client$ has a $d$-dimensional query vector $\MZ$. $\Client$ obtains $\funcML{linr}\big((\MW,\Vb), \MZ\big) = \MW \band \MZ + \Vb $, where $\MW \band \MZ$ denotes the dot-product of $\MW$ and $\MZ$.
	%----
	\item[$\circ$] {\bf SVM Regression}: $\Model$ holds $\{\alpha_j , y_j\}_{j=1}^k$, $d$-dimensional support vectors $\{\MX_j\}_{j=1}^k$ and bias $\Vb$, while $P_c$ holds a $d$-dimensional query  $\MZ$.  $\Client$ obtains $\funcML{svmr}\big((\{\alpha_j, y_j, \MX_j\}_{j=1}^k), \MZ\big)  = \sum_{j=1}^{k} \alpha_j y_j \allowbreak (\MX_j \band \MZ) + \Vb$. 
	%----
	\item[$\circ$] {\bf Logistic Regression}: The inputs of $\Model$ and $\Client$ are similar to linear regression.  $\Model$ needs to provide an additional input $t$  in the range $[0, 1]$.  $\Client$ obtains $\funcML{logr}\big((\MW,\Vb,t), \MZ\big) = \sign((\MW \band \MZ + \Vb) - \ln~(\frac{t}{1-t}))$, where $\sign(\cdot)$ returns the sign bit of its argument.  Since the values are represented in $2$'s complement representation, $\sign()$  returns the most significant bit (MSB) of its argument. 
	%---
	\item[$\circ$] {\bf SVM Classification}: The inputs of $\Model$ and $\Client$ remain the same as in SVM regression. But the output to $\Client$ changes to $\funcML{svmc}\big((\{\alpha_j,\allowbreak y_j, \MX_j\}_{j=1}^k), \MZ\big)  = \sign(\sum_{j=1}^{k} \alpha_j y_j \allowbreak (\MX_j \band \MZ) + \Vb)$. 
	%----
\end{myitemize}
%----
%\clearpage
\section{Implementation and Benchmarking}
\label{sec:Implementation}
In this section, we provide empirical results for our 3PC and secure prediction protocols. We start with the description of the setup environment-- software, hardware, and network.

%----------------------------------------------------------------
\paragraph{Network \& Hardware Details} 
%----------------------------------------------------------------
We have experimented both  in a LAN (local)  and a WAN (cloud) setting.  In the LAN setting, our machines ($P_0, P_1, P_2$) are equipped with Intel Core i7-7790 CPU with 3.6 GHz processor speed and 32 GB RAM. In the WAN setting, we use Microsoft Azure Cloud Services with machines located in South East Asia ($P_0$), North Europe ($P_1$) and North Central US ($P_2$). We used Standard E4s v3 instances, where machines are equipped with 32 GB RAM and 4 vcpus. Every pair of parties are connected by bi-directional communication channels in both the LAN and WAN setting, facilitating simultaneous data exchange between them. We consider a LAN with $1$Gbps and a WAN with $25$Mbps channel bandwidth. We measured the average round-trip time ($\rtt$) for communicating 1 KB of data between  $P_0$-$P_1$, $P_1$-$P_2$ and $P_0$-$P_2$ in both the setting. In the LAN setting, the average $\rtt$ turned out to be $0.47 ms$. In the WAN setting, the $\rtt$ between $P_0$-$P_1$, $P_1$-$P_2$ and $P_0$-$P_2$ are $201.928 ms$, $81.736 ms$ and $229.792 ms$ respectively. We use a TCP-IP connection between each set of parties.

%----------------------------------------------------------------
\paragraph{Software Details} 
%----------------------------------------------------------------
Our code follows the standards of C++11. We implemented our protocols in both semi-honest and malicious setting, using ENCRYPTO library~\cite{ENCRYPTO}. We used SHA-256 to instantiate the hash function. We use multi-threading to facilitate efficient computation and communication among the parties. For benchmarking, we use the AES-128~\cite{AESBristol} circuit. For ML prediction, since the code for ABY3 \cite{MR18} was not available, we implemented their protocols in our framework for benchmarking. We run each experiment 20 times and report the average for our measurements. 

%----------------------------------------------------------------
\paragraph{Parameters for Comparison}
%----------------------------------------------------------------
All our constructions are compared against their closest competitors which are implemented in our environment for a fair comparison. We consider five parameters for comparison-- latency (calculated as the maximum of the runtime of the parties or servers in case of secure prediction) in both LAN and WAN, total communication complexity and throughput of the {\em online} phase over LAN and WAN.  For 3PC over LAN,  the throughput is calculated as the number of AES circuits that can be computed per second. As an AES evaluation takes more than a second in WAN, we change the notion of throughput in WAN to the number of AND gates that can be computed per second. For the case of secure prediction, throughput is taken as a number of queries that can be processed per second in LAN and per minute in WAN. For simplicity, we use {\em online throughput} to denote the throughput of the online phase. The discrepancy across the benchmarking parameters for LAN and WAN comes from the difference in $\rtt$ (order of microseconds for LAN and milliseconds for WAN).

%----------------------------------------------------------------
\subsection{Experimental Results} 
%---------------------------------------------------------------- 

%----------------------------------------------------------------
\subsubsection{Results for 3PC} 
%----------------------------------------------------------------    
In \tabref{comp3PC}, we compare our 3PCs over the boolean ring ($\Z{}$) both in semi-honest and malicious setting  with their closest competitors  \cite{AFLNO16} and \cite{ABFLLNOWW17} respectively  in terms of latency and communication.
%-------------------------------
\begin{table}[htb!]
	\resizebox{0.48\textwidth}{!}
	{
		\begin{tabular}{c|c|c|c|c|c|c|c}
			\toprule
			\multirow{2}{*}{Protocol} & \multirow{2}{*}{Work}  
			& \multicolumn{2}{c|}{LAN Latency $(ms)$}  & \multicolumn{2}{c|}{WAN Latency $(s)$} 
			& \multicolumn{2}{c}{Communication (KB)} \\ \cmidrule{3-8}
			&						  & Offline    & Online         		  & Offline    & Online  			  & Offline    & Online  \\
			\midrule     
			%----                              
			\multirow{2}[2]{*}{Semi-honest}  	  
			& \cite{AFLNO16}	
			& 0 	& 254.8 	
			& 0 	& 8.96 	
			& 0 	& 1.99\\ \cmidrule{2-8} 
			& {\bf This}  		
			& 0.48 	& 254.8 	
			& 0.23  & 3.19   	
			& 0.66 	& 1.33\\ \midrule
			\multirow{2}[2]{*}{Malicious}  
			& \cite{ABFLLNOWW17}			 	
			& 1.44 	& 260.72 	
			& 0.71 	& 9.42 	
			& 8.06 	& 6.06\\ \cmidrule{2-8}
			& {\bf This}  		
			& 2.37 	& 248.38 	
			& 0.88  	& 3.57 	
			& 10.72 & 2.69\\ 
			%---		
			\bottomrule
			%-----
		\end{tabular}
	}	
	\caption{Comparison of Our 3PC with  \cite{AFLNO16} and \cite{ABFLLNOWW17}\label{tab:comp3PC}}
\end{table}	
%-------------------------------

\vspace{-5mm}
Note that \tabref{comp3PC} does not include the runtime and communication for input-sharing and output-reconstruction phases. We provide the runtime and communication of our protocol for the aforementioned phases in \tabref{IpOpPhases}. For benchmarking, we let $P_0$ own 48 out of the 128 input wires of AES while $P_1$ and $P_2$ own 40 wires each. The table provides benchmarking for the fair reconstruction phase as well, which sees an increase  in the latency for the online phase due to increased round complexity.
%-------------------------------
\begin{table}[htb!]
	\resizebox{0.48\textwidth}{!}
	{
		\begin{tabular}{c|c|c|c|c|c|c|c}
			\toprule
			\multirow{2}{*}{Phase} & \multirow{2}{*}{Protocol}  
			& \multicolumn{2}{c|}{LAN Latency $(ms)$}  & \multicolumn{2}{c|}{WAN Latency $(s)$} 
			& \multicolumn{2}{c}{Comm. (KB)} \\ \cmidrule{3-8}
			&						  & Offline    & Online         		  & Offline    & Online  			  & Offline    & Online  \\
			\midrule     
			%----                              
			\multirow{2}[3]{*}{\makecell{Input\\Sharing}}  	  
			& Semi-honest	
			& 0 	& \multirow{2}[3]{*}{0.47}	
			& 0 	& \multirow{2}[3]{*}{0.23}	
			& 0.01 	& 0.02\\ 
			\cmidrule{2-3} \cmidrule{5-5} \cmidrule{7-8} 
			& Malicious  		
			& 0.47 	&  	
			& 0.23 	&   	
			& 0.02 	& 0.03\\ \midrule
			\multirow{2}[3]{*}{\makecell{Output\\Reconstruction}}  
			& Semi-honest			 	
			& \multirow{2}[3]{*}{0} & \multirow{2}[3]{*}{0.47}	
			& \multirow{2}[3]{*}{0} & \multirow{2}[3]{*}{0.23} 	
			& \multirow{2}[3]{*}{0}	& 0.05\\ 
			\cmidrule{2-2} \cmidrule{8-8}
			& Malicious  		
			&  		& 	
			&   	&  	
			& 		& 0.09\\  
			\midrule
			\makecell{Fair Output\\Reconstruction}
			& Malicious  		
			& 0.47	& 1.91 	
			& 0.23 	& 0.77 	
			& 0.25	& 0.19\\  
			%---		
			\bottomrule
			%-----
		\end{tabular}
	}	
	\caption{Benchmarking for Input Sharing and Output Reconstruction Phases of Our 3PC Protocol\label{tab:IpOpPhases}}
\end{table}	
%--------------------------------

\vspace{-4mm}
In the semi-honest setting, we observe that the online latency for \cite{AFLNO16} and our protocol remain same over LAN. This is because both protocols require the same number of rounds of interaction during the online phase and the $\rtt$ among every pair of parties remain the same. Over WAN, our protocol outperforms \cite{AFLNO16} in terms of online latency. We observe that this improvement comes from the asymmetry in the $\rtt$ among the parties. In detail, our protocol has {\em only} one pair amongst the three pairs of parties to communicate for most of the rounds in the online phase. Thus, when compared with existing protocols, we have an additional privilege where we can assign the roles of the parties effectively across the machines so that the pair of parties having the most communication in the online phase is assigned the lowest $\rtt$. As a result, the time taken by a single round of communication comes down to the  {\em minimum} of the $\rtt$s among all the pairs, as opposed to the {\em maximum}. Thus we achieve a gain of (maximum $\rtt$)/(minimum $\rtt$) in time {\em per} round of communication, compared to the existing protocols. 

%%%--------------------RTT Plot------------------------------------------
\begin{figure}[htb!]
	\resizebox{0.4\textwidth}{!}
	{
		%-----------------------------------------------------
		\begin{tikzpicture}
		\begin{axis}[legend pos=north west, xlabel={multiplicative depth in powers of 2}, ylabel={latency in $ms$}, cycle list name = exotic]
		%--------
		\addplot+[smooth,mark=*] plot coordinates { (7, 28.74) (8, 57.47) (9, 114.95) (10, 229.90) (11, 459.78) (12, 919.55)};
		\addlegendentry{\bf This}
		%--------
		\addplot+[smooth,mark=x] plot coordinates { (7, 29.21) (8, 57.94) (9, 115.42) (10, 230.37) (11, 460.25) (12, 920.02) };
		\addlegendentry{\cite{AFLNO16,ABFLLNOWW17}}
		%--------
		\end{axis}
		\node[align=center,font=\bfseries, xshift=2.5em, yshift=-2em] (title) at (current bounding box.north) {LAN};
		\end{tikzpicture}
	}
	
	\resizebox{0.4\textwidth}{!}
	{
		%-----------------------------------------------------
		\begin{tikzpicture}
		\begin{axis}[legend pos=north west, xlabel={multiplicative depth in powers of 2}, ylabel={latency in $s$}, cycle list name = exotic]
		%--------
		\addplot+[smooth,mark=*] plot coordinates { (7, 10.46) (8, 20.92) (9, 41.83) (10, 83.67) (11, 167.33) (12, 334.66)};
		\addlegendentry{\bf This}
		%--------
		\addplot+[smooth,mark=x] plot coordinates { (7, 29.4) (8, 58.8) (9, 117.61) (10, 235.22) (11, 470.44) (12, 940.87) };
		\addlegendentry{\cite{AFLNO16,ABFLLNOWW17}}
		%--------
		\end{axis}
		\node[align=center,font=\bfseries, xshift=2.5em, yshift=-2em] (title) at (current bounding box.north) {WAN};
		\end{tikzpicture}
		%-----------------------------------------------------
	}
	\vspace{-2mm}
	\caption{Plot of Online Latency against Multiplicative Depth for 3PC Protocols \label{fig:plot_latency}}
\end{figure}
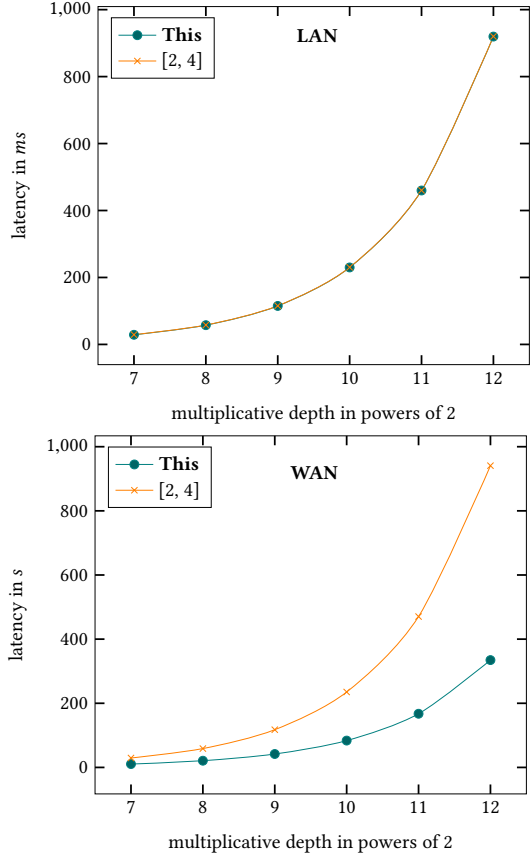

In \figref{plot_latency}, we compare the online latency of our protocols with their competitors, for a varying multiplicative depth (that dictates the round complexity). The same plot applies to both the semi-honest setting and malicious setting, as they differ by a single round and its impact vanishes with the growing number of rounds. It is clear from the plot that the impact of $\rtt$ becomes more visible with the increase in the number of online rounds, leading to improved efficiency. 

%-----------------------------------------------
\begin{table}[htb!] 
	\resizebox{0.47\textwidth}{!}
	{
		\begin{tabular}{c|c|c|c|c|c|c}
			\toprule
			\multirow{2}{*}{Setting} & \multicolumn{3}{c|}{Semi-honest} & \multicolumn{3}{c}{Malicious}   	\\ \cmidrule{2-7} 
			& \cite{AFLNO16} & {\bf This} & Improv. & \cite{ABFLLNOWW17} & {\bf This} & Improv. \\
			\midrule
			LAN  & 3296.7 	& 3296.7  & {1$\times$}        & 3221.85 & 3381.91 & {\bf 1.05$\times$} \\ \midrule 
			WAN  & 8.71~M 	& 13.1~M  & {\bf 1.51$\times$} & 2.9~M 	& 4.34~M	& {\bf 1.50$\times$} \\ 
			\bottomrule 
		\end{tabular}
	}
	\caption{Comparison of 3PC Online Throughput\label{tab:3pcTP}}
\end{table}
%-----------------------------------------------
\vspace{-5mm}

%-------------------------------
Now, we compare the online throughput for 3PC over both LAN (\#AES/sec) and WAN (\#AND/sec) setting  and the results appear in \tabref{3pcTP} (`M' denotes million and `Improv.' denotes improvement). \tabref{3pcTP} shows that our protocol's online throughput is clearly better than that of its competitors. This is mainly because of the improvement in online communication, though the asymmetry in our protocol has a contribution to it. In the semi-honest setting, our protocol is able to effectively push around $33\%$ of the total communication to the offline phase, resulting in an improved online phase. In the malicious setting, our protocol reduces online communication by a factor of $2.25\times$ with an increase in the offline phase by a factor of $1.75\times$, when compared with the state-of-the-art protocols. 

%----------------------------------------------------------------
\subsubsection{Results for Secure Prediction}
%----------------------------------------------------------------
We benchmark  our ML protocols that cover regression functions (linear and SVM) and classification functions (logistic and SVM)  over a ring $\Z{64}$. We report our performance for MNIST database \cite{MNIST10} that has $d=784$ features and compare our results with ABY3 \cite{MR18} (with the removal of extra tools as mentioned in the introduction). The comparison of latency and communication appears below.
%----------------------------------------------------------------
\paragraph{Regression}
%----------------------------------------------------------------
For regression, the servers compute $\shrd$-shares of the function $\MW \band \MZ + \Vb$, given the $\shrd$-shares of $\shr{\MW}, \shr{\MZ}$ and $\shr{\Vb}$. This is computed by parties executing secure dot-product on $\shr{\MW}$ and $\shr{\MZ}$, followed by locally adding the result with $\shrd$-shares of $\Vb$. Here we provide benchmarking for two regression algorithms, namely Linear Regression and Linear SVM Regression. 
%---
Though the aforementioned algorithms serve a different purpose, we observe that their underlying computation is same from the viewpoint of the servers, apart from the values $\MW, \MZ$ and $\Vb$ being  different as mentioned in \secref{MLAlgo}. Thus we provide a single benchmark, capturing both the algorithms and the results appear in \tabref{MLRegRuntime}.
%-------------
\begin{table}[htb!]
	\resizebox{0.4\textwidth}{!}
	{
		\begin{tabular}{c|c|c|c|c|c}
			\toprule
			\multirow{2}[2]{*}{Setting}
			& \multirow{2}[2]{*}{Work}  
			& \multicolumn{2}{c|}{Semi-honest}  
			& \multicolumn{2}{c}{Malicious}\\ \cmidrule{3-6}
			& & Offline  & Online &   Offline    & Online\\
			\midrule     
			%---- 
			\multirow{2}[2]{*}{\makecell{LAN\\($ms$)}} 
			& ABY3 			& 0    & 0.62 & 1.61 & 1.56\\ \cmidrule{2-6}
			& {\bf This} 	& 0.52 & 0.61 & 2.56 & 1.07\\ \midrule
			%----
			\multirow{2}[2]{*}{\makecell{WAN\\($s$)}} 
			& ABY3 			& 0    & 0.23 & 0.72 & 0.70\\ \cmidrule{2-6}
			& {\bf This} 	& 0.23 & 0.09 & 1.1 & 0.44\\ \midrule
			%----
			\multirow{2}[2]{*}{\makecell{Comm.\\(KB)}}
			& ABY3 			& 0    & 0.02 & 73.5    & 55.13\\ \cmidrule{2-6}
			& {\bf This} 	& 0.01 & 0.01 & 128.63  & 12.27\\ 
			%---		
			\bottomrule
			%-----
		\end{tabular}
	}
	\caption{Comparison of Latency and Communication for Regression Protocols\label{tab:MLRegRuntime}}
\end{table}	
%-------------------------------
\vspace{-5mm}

In the semi-honest setting, similar online latency for both protocols over LAN can be justified by the similar $\rtt$ among parties. Over WAN, the asymmetry in the $\rtt$ among the parties (as mentioned for the case of 3PC) adds benefit to our protocol. In the malicious setting, the result is further improved, since we require one less round when compared with ABY3 in the online phase.    
%-----------------------------------------------
\begin{table}[htb!] 
	\resizebox{0.47\textwidth}{!}
	{
		\begin{tabular}{c|c|c|c|c|c|c}
			\toprule
			\multirow{2}{*}{Setting} & \multicolumn{3}{c|}{Semi-honest} & \multicolumn{3}{c}{Malicious}   	\\ \cmidrule{2-7} 
			& ABY3 & {\bf This} & Improv. & ABY3 & {\bf This} & Improv. \\
			\midrule
			LAN  & 0.645~M 	& 0.656~M  & {\bf1.02$\times$}  & 0.007~M & 0.010~M & {\bf 1.5$\times$} \\ \midrule 
			WAN  & 0.104~M 	& 0.267~M  & {\bf 2.56$\times$} & 0.010~M & 0.016~M & {\bf 1.5$\times$} \\ 
			\bottomrule 
		\end{tabular}
	}
	\caption{Online Throughput of Regression Protocols\label{tab:MLRegTP}}
\end{table}
%-----------------------------------------------
\vspace{-5mm}

We now provide an {\em online} throughput comparison of our regression protocols over LAN (queries/sec) and  WAN (queries/min) setting and the result appear in \tabref{MLRegTP}. We observe that the throughput was further boosted in the malicious setting because of our efficient dot-product protocol (\secref{ProtML}) with which we could improve the online communication by a factor of $4.5\times$ when compared to ABY3.

In \figref{plot_MLRegTP}, we present a comparison of online throughput (\#queries/sec for LAN and \#queries/min for WAN) against the number of features in the malicious setting, for a number of features varying from 500 to 2500. Since the online communication cost is independent of the feature size in the semi-honest setting, we omit to plot the same. The plot clearly shows that our protocol for regression outperforms ABY3 in terms of online throughput. The reduction in throughput with the increase in feature size for both ours as well as ABY3's can be explained with the increase in communication for higher feature sizes.
\vspace{-2mm}
%%%-----------------RTT Plot ----------------------------------------
\begin{figure}[htb!]
	\resizebox{0.4\textwidth}{!}
	{
		\begin{tikzpicture}
		\begin{axis}[legend pos=north east, xlabel={\# features}, ylabel={\# queries/sec}, cycle list name = exotic]
		%--------
		\addplot+[smooth,mark=*] plot coordinates { (500, 16384) (1000, 8192) (1500, 5461.33) (2000, 4096) (2500, 3276.8)};
		\addlegendentry{\bf This}
		%--------
		\addplot+[smooth,mark=x] plot coordinates { (500, 10922.67) (1000, 5461.33) (1500, 3640.89) (2000, 2730.67) (2500, 2184.53)};
		\addlegendentry{ABY3}
		%--------
		\end{axis}
		\node[align=center,font=\bfseries, xshift=2.5em, yshift=-2em] (title) at (current bounding box.north) {LAN};
		\end{tikzpicture}
	}
	\resizebox{0.4\textwidth}{!}
	{
		\begin{tikzpicture}
		\begin{axis}[legend pos=north east, xlabel={\# features}, ylabel={\# queries/min}, cycle list name = exotic]
		%--------
		\addplot+[smooth,mark=*] plot coordinates { (500, 24576) (1000, 12288) (1500, 8192) (2000, 6144) (2500, 4915.2)};
		\addlegendentry{\bf This}
		%--------
		\addplot+[smooth,mark=x] plot coordinates { (500, 16384) (1000, 8192) (1500, 5461.33) (2000, 4096) (2500, 3276.80)};
		\addlegendentry{ABY3}
		%--------
		\end{axis}
		\node[align=center,font=\bfseries, xshift=2.5em, yshift=-2em] (title) at (current bounding box.north) {WAN};
		\end{tikzpicture}
	}
	\vspace{-2mm}
	\caption{Plot of Online Throughput against Multiplicative Depth for Regression Protocols\label{fig:plot_MLRegTP}}
\end{figure}
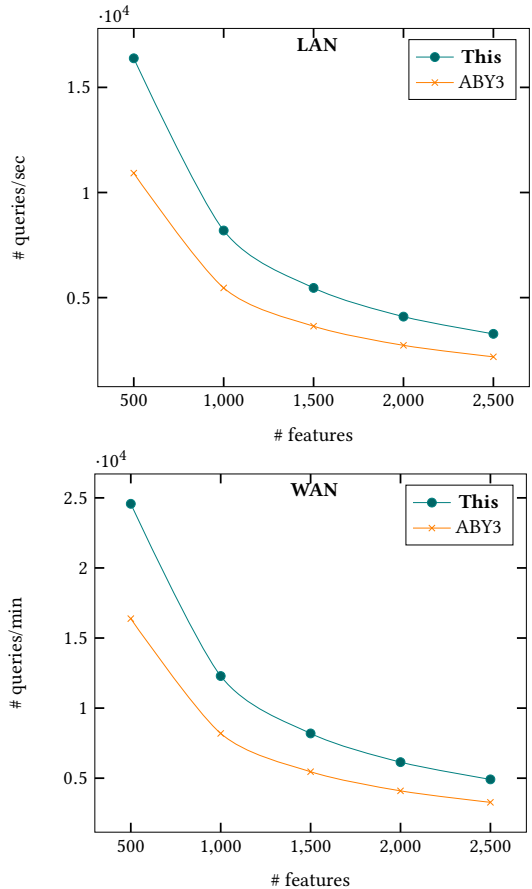

\vspace{-5mm}
%----------------------------------------------------------------
\paragraph{Classification}
For classification, the servers compute $\shrB{\cdot}$-shares of the function $\sign(\MW \band \MZ + \Vb)$, given the $\shrd$-shares of $\shr{\MW}, \shr{\MZ}$ and $\shr{\Vb}$. Towards this, parties first execute secure dot-product on $\shr{\MW}$ and $\shr{\MZ}$, followed by locally adding the result with $\shr{\Vb}$. Then parties execute secure comparison protocol on the result obtained from the previous step to generate the boolean share of $\sign(\MW \band \MZ + \Vb)$.
Here we consider two classification algorithms, namely Logistic Regression and Linear SVM Classification. Similar to the case with Regression, both algorithms share the same computation from the server's perspective and thus we provide a single benchmark. The results appear in \tabref{MLClassRuntime} and the online throughput comparison appears in \tabref{MLClassTP}.
%\vspace{-2mm}

%----------------------------------------------------------------
\begin{table}[H]
	\resizebox{0.39\textwidth}{!}
	{
		\begin{tabular}{c|c|c|c|c|c}
			\toprule
			\multirow{2}[2]{*}{Setting}
			& \multirow{2}[2]{*}{Work}  
			& \multicolumn{2}{c|}{Semi-honest}  
			& \multicolumn{2}{c}{Malicious}\\ \cmidrule{3-6}
			& & Offline  & Online &   Offline    & Online\\
			\midrule     
			%---- 
			\multirow{2}[2]{*}{\makecell{LAN\\($ms$)}} 
			& ABY3 			& 0    & 3.48 & 1.63 & 4.42\\ \cmidrule{2-6}
			& {\bf This} 	& 0.54 & 1.58 & 2.57 & 2.53\\ \midrule
			%----
			\multirow{2}[2]{*}{\makecell{WAN\\($s$)}}  
			& ABY3 			& 0    & 1.61 & 0.72 & 2.08\\ \cmidrule{2-6}
			& {\bf This} 	& 0.23 & 0.55 & 1.1  & 0.98\\ \midrule
			%----
			\multirow{2}[2]{*}{\makecell{Comm.\\(KB)}}
			& ABY3 			& 0    & 0.07 & 73.7 & 55.3\\ \cmidrule{2-6}
			& {\bf This} 	& 0.01 & 0.04 & 129  & 12.4\\ 
			%---		
			\bottomrule
			%-----
		\end{tabular}
	}
	\caption{Comparison of Latency and Communication for Classification Protocols\label{tab:MLClassRuntime}}
\end{table}	
%-------------------------------
\vspace{-6mm}

%-----------------------------------------------
\begin{table}[htb!] 
	\resizebox{0.47\textwidth}{!}
	{
		\begin{tabular}{c|c|c|c|c|c|c}
			\toprule
			\multirow{2}{*}{Setting} & \multicolumn{3}{c|}{Semi-honest} & \multicolumn{3}{c}{Malicious}   	\\ \cmidrule{2-7} 
			& ABY3 & {\bf This} & Improv. & ABY3 & {\bf This} & Improv. \\
			\midrule
			LAN  & 0.115~M 	& 0.253~M & {\bf2.2$\times$}   & 0.007~M & 0.010~M  & {\bf 1.5$\times$} \\ \midrule 
			WAN  & 0.015~M 	& 0.044~M & {\bf 2.93$\times$} & 0.010~M & 0.016~M	& {\bf 1.5$\times$} \\ 
			\bottomrule 
		\end{tabular}
	}
	\caption{Online Throughput of Classification Protocols\label{tab:MLClassTP}}
\end{table}
%-----------------------------------------------
\vspace{-5mm}

%%%--------------------------------RTT Plot--------------------------------
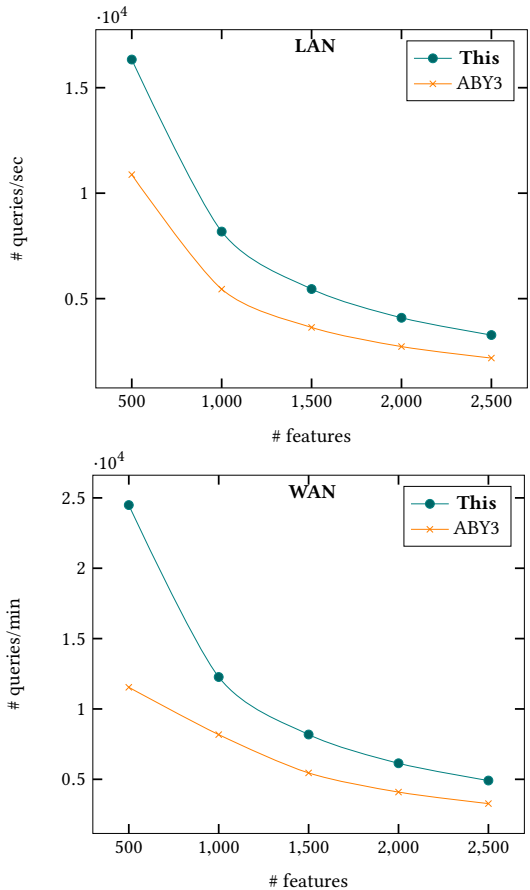
\begin{figure}[htb!]
	\resizebox{0.4\textwidth}{!}
	{
		\begin{tikzpicture}
		\begin{axis}[legend pos=north east, xlabel={\# features}, ylabel={\# queries/sec}, cycle list name = exotic]
		%--------
		\addplot+[smooth,mark=*] plot coordinates { (500, 16335) (1000, 8179.73) (1500, 5455.88) (2000, 4092.93) (2500, 3274.84)};
		\addlegendentry{\bf This}
		%--------
		\addplot+[smooth,mark=x] plot coordinates { (500, 10879.15) (1000, 5450.43) (1500, 3636.04) (2000, 2727.94) (2500, 2182.79)};
		\addlegendentry{ABY3}
		%--------
		\end{axis}
		\node[align=center,font=\bfseries, xshift=2.5em, yshift=-2em] (title) at (current bounding box.north) {LAN};
		\end{tikzpicture}
	}
	\resizebox{0.4\textwidth}{!}
	{
		\begin{tikzpicture}
		\begin{axis}[legend pos=north east, xlabel={\# features}, ylabel={\# queries/min}, cycle list name = exotic]
		%--------
		\addplot+[smooth,mark=*] plot coordinates { (500, 24489.8) (1000, 12269.6) (1500, 8183.82) (2000, 6139.4) (2500, 4912.25)};
		\addlegendentry{\bf This}
		%--------
		\addplot+[smooth,mark=x] plot coordinates { (500, 11538.46) (1000, 8175.65) (1500, 5454.06) (2000, 4091.91) (2500, 3274.18)};
		\addlegendentry{ABY3}
		%--------
		\end{axis}
		\node[align=center,font=\bfseries, xshift=2.5em, yshift=-2em] (title) at (current bounding box.north) {WAN};
		\end{tikzpicture}
	}
	\vspace{-2mm}
	\caption{Plot of Online Throughput against Multiplicative Depth for Classification Protocols\label{fig:plot_MLClassTP}}
\end{figure}

In this case, we observe that our protocol outperforms ABY3 in all the settings. This is mainly due to our Secure Comparison protocol (\secref{ProtML}) where we improve upon both communication and rounds in the online phase. The effect of this improvement becomes more visible for applications where the secure comparison is used extensively. Similar to Regression, in \figref{plot_MLClassTP}, we provide below a comparison of online throughput (\#queries/sec for LAN and \#queries/min for WAN) against the number of features in the malicious setting.

%----------------------------------------------------------------
\subsection{Restricted Bandwidth Setting}
\label{sec:BWEffect}
%----------------------------------------------------------------
We observe that the asymmetry of our constructions further comes to our advantage for throughput. That is, while a drop in bandwidth between any pair of parties significantly affects the throughput of  the existing protocols, the throughput of ours does not get affected much as long as the drop occurs between the pair(s) of parties handling a low volume of data. The purpose of this setting is to show that for setups with varying bandwidths among the servers, our protocol has an advantage in choosing the roles of the servers whereas existing works cannot.
%%%--------------------------------RTT Plot--------------------------------
\begin{figure}[htb!]
	\resizebox{0.4\textwidth}{!}
	{
		\begin{tikzpicture}
		\begin{axis}[legend pos=north east, xlabel={\# features}, ylabel={\# queries/sec}, cycle list name = exotic]
		%--------
		\addplot+[smooth,mark=*] plot coordinates { (500, 16384) (1000, 8192) (1500, 5461.33) (2000, 4096) (2500, 3276.8)};
		\addlegendentry{\bf This}
		%--------
		\addplot+[smooth,mark=x] plot coordinates { (500, 1092.27) (1000, 546.13) (1500, 364.09) (2000, 273.07) (2500, 218.45)};
		\addlegendentry{ABY3}
		%--------
		\end{axis}
		\node[align=center,font=\bfseries, xshift=2.5em, yshift=-2em] (title) at (current bounding box.north) {LAN};
		\end{tikzpicture}
	}
	\resizebox{0.4\textwidth}{!}
	{
		\begin{tikzpicture}
		\begin{axis}[legend pos=north east, xlabel={\# features}, ylabel={\# queries/min}, cycle list name = exotic]
		%--------
		\addplot+[smooth,mark=*] plot coordinates { (500, 24489.8) (1000, 12269.6) (1500, 8183.82) (2000, 6139.4) (2500, 4912.25)};
		\addlegendentry{\bf This}
		%--------
		\addplot+[smooth,mark=x] plot coordinates { (500, 6527.49) (1000, 3270.26) (1500, 2181.62) (2000, 1636.76) (2500, 1309.67)};
		\addlegendentry{ABY3}
		%--------
		\end{axis}
		\node[align=center,font=\bfseries, xshift=2.5em, yshift=-2em] (title) at (current bounding box.north) {WAN};
		\end{tikzpicture}
	}
	\vspace{-2mm}
	\caption{Plot of Online Throughput against Multiplicative Depth for Classification Protocols  in the Malicious Setting under Restricted Bandwidth\label{fig:plot_MLClassTPBW}}
\end{figure}
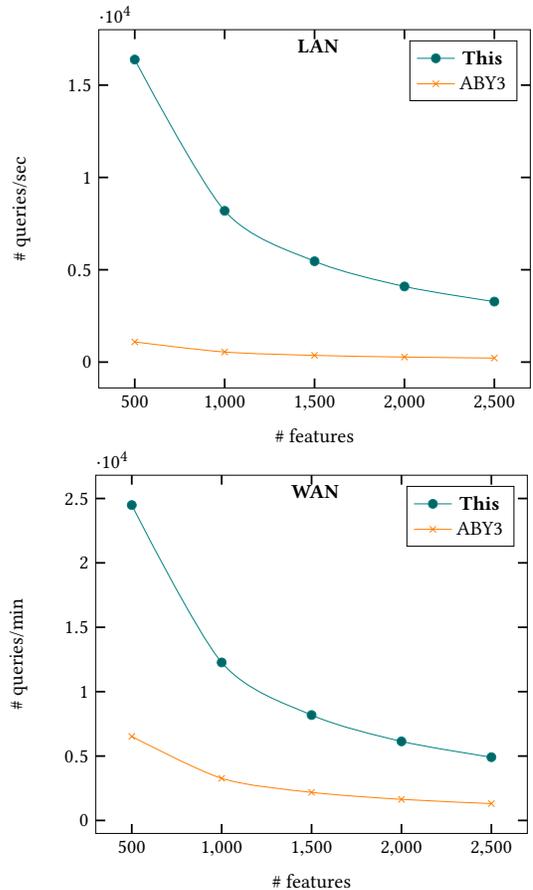

To demonstrate this positive impact, we test the throughput of our ML constructions in a modified network setting where the bandwidth between one of the pairs, namely $P_0$ and $P_2$ is restricted to 100Mbps (instead of 1Gbps) in LAN and 10Mbps (instead of 25Mbps) in WAN setting. This restriction significantly drops the throughput of the existing constructions as they need all the pairs to communicate equally, while ours remain unaffected. The cut-down on bandwidth  does not make any difference in latency (that is measured for {\em one} execution) and communication complexity. We provide a comparison of throughput in the malicious setting in \tabref{MLClassTPBW}. 
%-----------------------------------------------
\begin{table}[htb!] 
	\resizebox{0.48\textwidth}{!}
	{
		\begin{tabular}{c|c|c|c|c|c|c}
			\toprule
			\multirow{2}{*}{Setting} & \multicolumn{3}{c|}{Regression} & \multicolumn{3}{c}{Classification}   	\\ \cmidrule{2-7} 
			& ABY3 & {\bf This} & Improv. & ABY3 & {\bf This} & Improv. \\
			\midrule
			LAN  & 0.001~M 	& 0.010~M & {\bf15$\times$}    & 0.001~M & 0.010~M  & {\bf 15.01$\times$} \\ \midrule 
			WAN  & 0.004~M 	& 0.016~M & {\bf 3.75$\times$} & 0.004~M & 0.016~M	& {\bf 3.75$\times$} \\ 
			\bottomrule 
		\end{tabular}
	}
	\caption{Online Throughput of ML Protocols in the Malicious Setting under Restricted Bandwidth\label{tab:MLClassTPBW}}
\end{table}
%-----------------------------------------------
\vspace{-5mm}

The comparison of online throughput (\#queries/sec for LAN while \#queries/min for WAN) against the number of features in the malicious setting for classification protocols appear in \figref{plot_MLClassTPBW}.
%----
\section{Conclusions}
\label{sec:Conclusions}
In this work, we presented efficient protocols for the three party setting (3PC) tolerating at most one corruption. We applied our results in the domain of secure machine learning prediction for a range of functions -- Linear Regression, Linear SVM Regression, Logistic Regression, and Linear SVM classification. The theoretical improvements over the state-of-the-art protocols were backed up by an extensive benchmarking.

%----------------------------------------------------------------
\paragraph{\bf Open Problems}
%----------------------------------------------------------------
Our techniques are tailor-made for 3PC with 1 corruption. Extending these techniques to the case of an arbitrary $Q^{(2)}$ adversary structure~\cite{SmartW19} is left as an open problem.

%-------------------------------
\paragraph{\bf Acknowledgements}
%-------------------------------
We would like to thank Thomas Schneider for helpful discussions, comments, and pointers.

%Ashish Choudhury supported by R\& D work undertaken in the project under the Visvesvaraya Ph.D.  Scheme of  Ministry of Electronics \& Information  Technology, Government of India, being implemented by Digital India Corporation (Formerly Media Lab Asia). Arpita Patra supported by SERB Women Excellence Award 2017 (DSTO 1706).

%----
%\clearpage
\bibliographystyle{ACM-Reference-Format}
\bibliography{main}
\appendix
%------
%\clearpage
\section{Building blocks for malicious security}
\label{app:3PCAppendix}
%\input{Appendix_3PC}
%----------------------------------------------------------------
\subsection{Instantiating $\FTRIPLES$}
\label{app:FTRIPLES}
%----------------------------------------------------------------
Here, we present a protocol $\PiTripGen$ (\boxref{fig:TripGenII}) that  instantiate functionality  $\FTRIPLES$ over $\Z{\ell}$, inspired by the works of \cite{FLNW17,ABFLLNOWW17}. The techniques of \cite{FLNW17,ABFLLNOWW17} work for any underlying linear secret-sharing scheme.  We avoid the detailed security proof for $\PiTripGen$, which can be easily derived from \cite{FLNW17,ABFLLNOWW17}. We begin with a sub-protocol $\PiRand$, used in $\PiTripGen$. Protocol $\PiRand$ allows the parties to generate a random and private $\shrd$-shared value $\val$. Towards this, parties $P_0,P_1$ locally sample $\PadA{\val}$, $P_0,P_2$ sample $\PadB{\val}$ while parties $P_1,P_2$ sample $\Mask{\val}$. The value $\val$ is defined as $\val = \Mask{\val} - \PadA{\val} - \PadB{\val}$.

%-------------
\begin{myprotocolbox}{}{Protocol to generate $N$ random and private $\shrd$-shared multiplication triples}{fig:TripGenII}
	%-----
	\justify 
	If a party obtains $\bot$ during any stage of the protocol or did not receive an expected message, then it outputs $\bot$ and $\abort$.
	%----	
	\begin{myitemize}
		%--
		\item[--] {\bf Generating Multiplication Triples Optimistically}: Let $M = BN+C$. The parties execute $2M$ instances of $\PiRand$ to generate $\{(\shr{\wmd_k}, \newline \shr{\wme_k}) \}_{k = 1, \ldots, M}$. For $k = 1, \ldots, BN + C$, the parties execute $\PiMultS$ on $\shr{\wmd_k}$ and $\shr{\wme_k}$ to obtain $\shr{\wmf_k}$. Let $\vec{D} = [(\shr{\wmd_k}, \shr{\wme_k}, \shr{\wmf_k})]_{k = 1, \ldots, M}$.
		%---		
		\item[--] {\bf Cut and Bucket}: Here the parties perform the first verification by opening
		$C$ triples, and then randomly divide the remainder into buckets as follows.
		%---- 
		\begin{myitemize}
			%-----
			\item[--] The parties generate a random permutation $\pi$ over $\{1, \ldots, BN+C \}$ and permute the elements of $\vec{D}$ according to $\pi$.
			%----
			\item[--] The parties publicly reconstruct each of the first $C$ triples in $\vec{D}$
			(by executing $\PiRecS(\shrd, \Partyset)$ and output $\bot$, if any of these $C$ triples is not a multiplication triple.
			%--- 
			\item[--] The remaining $BN$  triples in $\vec{D}$ are arranged into buckets $B_1, \ldots,\newline B_N$, each containing $B$ triples. 
			%---
		\end{myitemize}
		%-----------	
		\item[--] {\bf Check Buckets}: The parties initialize a vector $\vec{d}$ of length $N$. Then, for $k = 1, \ldots, N$, the parties do the following:
		%------		
		\begin{myitemize}	 
			%------
			\item[--] Let $\{ (\shr{\wmd_{k, j}}, \shr{\wme_{k, j}}, \shr{\wmf_{k, j}} ) \}_{j = 1, \ldots, B}$ denote the $B$ shared triples in the bucket $B_k$.
			%---- 
			\item[--] For $j = 2, \ldots, B$, the parties execute $\PiTripCheck$ on $( \shr{\wmd_{k, 1}},   \allowbreak  \shr{\wme_{k, 1}}, \shr{\wmf_{k, 1}})$ and $( \shr{\wmd_{k, j}}, \shr{\wme_{k, j}},   \shr{\wmf_{k, j}})$.
			%---
			\item[--] The parties set $( \shr{\wmd_{k, 1}}, \shr{\wme_{k, 1}}, \shr{\wmf_{k, 1}})$ as the $k$th entry of $\vec{d}$.
			%----
		\end{myitemize}
		%------    
	\end{myitemize}
	%-----
	The parties output $\vec{d}$.
\end{myprotocolbox}
%----

Following the technique of  \cite{FLNW17},  protocol $\PiTripGen$ generates  $N$ independent $\shrd$-shared random and private multiplication triplets over $\Z{\ell}$ at one go. Informally, the parties first optimistically generate $BN+C$ shared random triples, followed by deploying the cut-and-choose technique. Namely $C$ triples from the set of $BN+C$ triples are randomly selected and opened to check if they are multiplication triples. The remaining $BN$  triples are randomly grouped into $N$ buckets, each containing $B$ triples. In each bucket, parties check if the first triple is a multiplication triple without opening it using the protocol $\PiTripCheck$ (\boxref{fig:PiCheck}), by deploying the remaining $B - 1$ triples in the bucket, one by one. If any of these verifications fail, then the parties abort, else they consider the first  triple in each of the $N$ buckets as the final output. Following \cite{FLNW17}, it follows that except with an error probability of at most $\frac{1}{N^{B-1}}$, if any of the $N$ output triplets is not a multiplication triplet, then the honest parties abort the protocol. 

In their follow-up work \cite{ABFLLNOWW17}, the authors have shown how to reduce the error probability of cut-and-choose technique  from $\frac{1}{N^{B-1}}$ to $\frac{1}{N^B}$, thus reducing the bucket size $B$ to $\frac{s}{\log_2{N}}$ to attain a statistical-security of $2^{-s}$. The idea behind their improvement is as follows: if the array of multiplication triples from the offline phase is randomly shuffled {\it after} all multiplication gates are evaluated (optimistically), then adversary can successfully cheat only if the random shuffle happens to match correct triples with correctly evaluated multiplication gates and incorrect triples with incorrectly evaluated multiplication gates. 

We observe that the above modification is applicable in our context as well. Following \cite{ABFLLNOWW17}, the parties can postpone verification of offline step of all the instances of $\PiMultM$. Once the offline step of all the instances of $\PiMultM$ corresponding to all the multiplication gates in the circuit is executed, the parties can randomly shuffle the set of triples. The parties can then use the $i$th triple from the reshuffled set to perform the pending verification corresponding to the offline step of the $i$th instance of $\PiMultM$. Notice that unlike \cite{ABFLLNOWW17}, in our context, the reshuffling of the set of triples happens in the offline phase itself. Excluding the cost of generating the random permutation $\pi$ in the protocol of \boxref{fig:TripGenII}, the amortized cost of generating a single multiplication triple will be as follows: there will be $2B$ instances of $\PiRand$ followed by $B$ instances of $\PiMultS$, followed by $B-1$ instances of $\PiTripCheck$. 

%----------------------------------------------------------------
\subsection{Properties of $\PiTripCheck$}
\label{app:PiTripCheck}
%----------------------------------------------------------------
%-------------
\begin{lemma}[Correctness \cite{CP17,FLNW17}]
	\label{lemma:TripCheck}
	%----
	Let $(\shr{\wmd}, \shr{\wme}, \shr{\wmf})$ be $\shrd$ sharing of random and private values $\wmd, \wme$ and $\wmf$, such that $\wmf = \wmd \wme$. Moreover, let $(\shr{\wma}, \shr{\wmb}, \shr{\wmc})$ be $\shrd$ sharing of $\wma, \wmb$ and $\wmc$, such that $\wmc = \wma \wmb + \Delta$, where $\Delta \in \Z{\ell}$. Then the following hold in $\PiTripCheck$:  If $\Delta \neq 0$, then every honest $P_i$ outputs $\bot$.
\end{lemma}
%-----
\begin{proof}
	%---
	In $\PiTripCheck$, during the reconstruction of $\rho$, $\sigma$ and $\tau$, protocol $\PiRecM$ ensures that no two honest parties output two different non-$\bot$ values. Now, in order to show the correctness, it suffices to show that $\tau = 0$ iff $\Delta = 0$. Note that,
	\vspace{-2mm}
	\begin{align*}
		\tau &= \wmc - \wmf - \sigma \wmd - \rho \wme - \sigma \rho\\
		&= \wmc - \wmd \wme - (\wmb-\wme) \wmd - (\wma-\wmd) \wme - (\wmb-\wme)(\wma-\wmd)\\
		&= \wmc - \wma \wmb = \Delta
	\end{align*}
	It is straightforward from the protocol step that every honest party outputs $\bot$ if $\Delta \ne 0$.
	%---
\end{proof}
%----------   
The privacy of $\PiTripCheck$ requires it to maintain the privacy of $\wma, \wmb$ and $\wmc$. Note that the values $\rho$ and $\sigma$ reveal nothing about $\wma$ and $\wmb$, as $\wmd, \wme$ are random and private. The privacy now follows since $\tau = \Delta$ and independent of $\wma, \wmb$ and $\wmc$.
%----
\section{3PC with semi-honest security}
\label{app:3PCSemi}
Here we prove that $\PiSemi$ securely realizes the standard ideal-world functionality $\FTHREEMPC$ (\boxref{fig:FTHREEMPC}) for securely evaluating any arithmetic circuit over $\Z{\ell}$. Our proof works in the $\FSETUP$-hybrid model.  
%-----------------
\begin{mysystembox}{}{Functionality $\FSETUP$ (semi-honest)}{fig:FSETUPSemi}
	$\FSETUP$ interacts with the parties in $\Partyset$ and the adversary $\Sim$ who may corrupt one of the parties.
	%---
	\justify $\FSETUP$ picks random keys $k_{01}, k_{02}, k_{12}, k_{\Partyset} \in \{0,1\}^{\csec}$ and sends $(k_{01}, k_{02})$, $(k_{01}, k_{12})$ and $(k_{02}, k_{12})$ to $P_0, P_1$ and $P_2$ respectively. In addition, $\FSETUP$ sends $k_{\Partyset}$ to all the parties.
	%---
\end{mysystembox}
%---------
We first consider the simple case, when $P_0$ is corrupted. Intuitively, the security follows from the fact, that $P_0$ does not see the messages exchanged between $P_1, P_2$ during the online phase, who actually perform the circuit-evaluation. So in essence, this is equivalent to $P_1, P_2$ using the preprocessing done by a trusted third party to do the circuit-evaluation (in the semi-honest setting, even a corrupt $P_0$ will do the pre-processing honestly).
%------------
\begin{theorem}
	\label{lemma:PiSemiCaseI}
	%-----------
	Protocol $\PiSemi$ securely realizes the functionality $\FTHREEMPC$ against a static, semi-honest adversary $\Adv$ in the $\FSETUP$-hybrid model, who corrupts $P_0$.
	%----
\end{theorem}
%-----
\begin{proof}
	%-----
	Let $\Adv$ be a real-world semi-honest adversary corrupting the distributor $P_0$ during the protocol $\PiSemi$. We present an ideal-world adversary (simulator) $\SPiSemi$ for $\Adv$ in \boxref{fig:SPiSemiCaseI} that simulates messages for corrupt $P_0$. The only communication to $P_0$ is during the output-reconstruction stage in the online phase. $\SPiSemi$ can easily simulate these messages, with the knowledge of function output and the masks corresponding to the circuit-output wires.
	%--------------------------
	\begin{mysimulatorbox}{}{Simulator $\SPiSemi$ for the case of corrupt $P_0$}{fig:SPiSemiCaseI}
		%----------
		\justify The simulator plays the role of the honest parties $P_1, P_2$ and  simulates  each step of $\PiSemi$ to corrupt $P_0$ as follows and finally outputs $\Adv$'s output. 
		%
		%----------------
		\justify \algoHead{Offline Phase:} $\SPiSemi$ emulates $\FSETUP$ and gives $k_{01}$, $k_{02}$ and $k_{\Partyset}$ to $P_0$. In addition, $\SPiSemi$ on behalf of $P_2$ receives $\GammaxyB$ from $\Adv$ for every multiplication gate $\gate = (\Wxyz)$. From these, it learns the $\pad$-masks for all the wires in $\ckt$. 
		
		%----
		\justify \algoHead{Online Phase:} On input $\{\wx_i\}$'s, the inputs of corrupt $P_0$ and  the function output $(\wy_1, \ldots, \wy_{\OS})$, $\SPiSemi$  simulates the output-reconstruction stage to $\Adv$ as follows. For every $\wy_j$, it computes $\Mask{{\wy}_j} = \wy_j + \Pad{{\wy}_j}$ and sends it to $\Adv$, on the behalf of $P_1$. Here $\Pad{{\wy}_j}$ is the mask corresponding to the output $\wy_j$ which $\SPiSemi$ can compute since he learns the entire $\pad$-masks during the offline phase.
		%----------
	\end{mysimulatorbox}
	%----------
	
	The proof now simply follows from the fact that simulated view and real-world view of the adversary are computationally indistinguishable.
	%-----
\end{proof}
%----------

We next consider the case, when the adversary corrupts one of the evaluators. Without loss of generality, we consider the case of a corrupt $P_1$ and the case of a corrupt $P_2$ is handled symmetrically. Intuitively, the security, in this case, follows from the fact that each $\pad$-mask is random (from the properties of the underlying PRF) and the one share that is learned by corrupt $P_1$ for each mask leaks nothing about them  and hence the masked values  reveal no additional information about the actual values over the wires.

%------ 
\begin{theorem}
	\label{lemma:PiSemiCaseII}
	%-----------
	Protocol $\PiSemi$ securely realizes the functionality $\FTHREEMPC$ against a static, semi-honest adversary $\Adv$ in the $\FSETUP$-hybrid model, who corrupts $P_1$ (and similarly $P_2$).
	%----
\end{theorem}
%-----
%-----
\begin{proof}
	%-----
	Let $\Adv$ be a real-world semi-honest adversary corrupting the evaluator $P_1$ during the protocol $\PiSemi$. We now present the steps of the ideal-world adversary (simulator) $\SPiSemi$ for $\Adv$ for this case in \boxref{fig:SPiSemiCaseII}. At a  high level, $\SPiSemi$  itself does the honest pre-processing on the behalf of $P_0$ and will simulate the entire circuit-evaluation, assuming the circuit-inputs of $P_0$ and $P_2$ to be $0$. In the output-reconstruction stage, it ``adjusts" the shares of circuit-output values on the behalf of $P_2$ so that $\Adv$ sees the same function output as in the real-world protocol.
	%--------------------------
	\begin{mysimulatorbox}{}{Simulator $\SPiSemi$ for the case of corrupt $P_1$}{fig:SPiSemiCaseII}
		%----------
		\justify The simulator plays the role of the honest parties $P_0, P_2$ and  simulates  each step of the protocol $\PiSemi$ to corrupt $P_1$ as follows  and finally outputs $\Adv$'s output.
		\justify \algoHead{Offline Phase:}  $\SPiSemi$ emulates $\FSETUP$ and gives $k_{01}$, $k_{12}$ and $k_{\Partyset}$ to $P_1$. $\SPiSemi$ chooses a random key $k_{02}$. With these, $\SPiSemi$, on the behalf of $P_0$, executes the offline steps of the instances of  $\PiShS$, $ \PiMultS$ and $\PiAdd$ for circuit-inputs, multiplication and addition gates respectively. In the process, it learns the masks for each wire in the $\ckt$  and $\gamma$-values for each multiplication gate.
		%
		%--    
		\justify \algoHead{Online Phase:}~ 
		%---
		\begin{myitemize}
			%---         
			\item[--] \justify {\em Sharing Circuit-input Values}:   For every circuit-input $\wx_j$ that $P_0$  inputs, $\SPiSemi$ sets $\wx_j = 0$ and simulates the messages of $P_0$ as part of the online steps of  $\PiShS(P_0, \wx_j)$.The inputs owned by $P_2$ are simulated similarly.
			%----
			\item[--] {\em Gate Evaluation}:  The simulator simulates the evaluation of each gate $\gate$ according to the topological order. No simulation is needed for an addition gate. If $\gate$ is a multiplication gate, then the simulator simulates the messages of $P_2$ as part of the online steps of the corresponding instance of $\PiMultS$.
			%----     
			\item[--] {\em Output Reconstruction}:  For $j = 1, \ldots, \OS$ let $[\Pad{{\wy}_j}] = (\PadA{{\wy}_j}, \PadB{{\wy}_j})$ be the sharing, available with the simulator and let $\Mask{{\wy}_j}$ be the simulated masked output, corresponding to $\wy_j$, available with $\Adv$. On input $\{\wx_i\}$'s, the inputs of corrupt $P_1$ and  the function output $(\wy_1, \ldots, \wy_{\OS})$, as part of online steps of the instance  $\PiRecS(\sqr{\Pad{{\wy}_j}})$, the simulator sends $\Mask{{\wy}_j} - \PadA{{\wy}_j} - {\wy}_j$ as the share of $\Pad{{\wy}_j}$, on the behalf of $P_2$ to $\Adv$.
			%---
		\end{myitemize}
	\end{mysimulatorbox}
	%-----
	
	It is easy to see that the simulated view and the real-world view of the adversary are computationally indistinguishable.	
	%----
\end{proof}
%------
%----
\section{3PC with malicious security}
\label{app:3PCMal}
Here we prove that $\PiMal$ securely realizes the standard ideal-world functionality $\FTHREEMPCABORT$ (\boxref{fig:FTHREEMPCABORT}) for securely evaluating any arithmetic circuit over $\Z{\ell}$ with selective abort. Our proof works in $\{\FSETUP,\FTRIPLES\}$-hybrid model. 

%---------
\begin{mysystembox}{}{Functionality $\FSETUP$ (malicious)}{fig:FSETUPMal}
	$\FSETUP$ interacts with the parties in $\Partyset$ and the adversary $\Sim$. $\FSETUP$ picks random keys $k_{01}, k_{02}, k_{12}, k_{\Partyset} \in \{0,1\}^{\csec}$.
	%---
	\begin{description}
		%----
		\item {\bf Output to adversary: }  If $\Sim$ sends $\abort$, then send $(\OUTPUT, \bot)$ to all the parties. Otherwise,  send $(\OUTPUT, \wy_i)$ to the adversary $\Sim$, where $\wy_i = (k_{i1}, k_{i2}, k_{\Partyset})$ when $P_0$ is corrupt and $\wy_i = (k_{0i}, k_{12}, \allowbreak k_{\Partyset})$ when $P_i \in \{P_1, P_2\}$ is corrupt.
		%--- 
		\item {\bf Output to selected honest parties: } Receive $(\SELECT, \{I\})$ from adversary $\Sim$, where $\{I\}$ denotes a subset of the honest parties. If an honest party $P_i$ belongs to $I$, send $(\OUTPUT, \bot)$, else send $(\OUTPUT, \wy_i)$. Here $\wy_i = (k_{i1}, k_{i2}, k_{\Partyset})$ when $P_i = P_0$ and $\wy_i = (k_{0i}, k_{12}, k_{\Partyset})$ when $P_i \in \{P_1, P_2\}$
		%---
	\end{description}
	%----
\end{mysystembox}
%---------
Since the protocol $\PiMal$ differs from $\PiSemi$ mainly in three protocols -- sharing ($\PiShM$), reconstruction ($\PiRecM$) and multiplication ($\PiMultM$) protocols, we provide the details of simulation for the same. 
%---------------------------------------------------------------------------------------------------------
We begin with the case, when $P_0$ is corrupted. 
%------------
\begin{theorem}
	\label{lemma:PiMalCaseI}
	%-----------
	In  $\{\FSETUP,\FTRIPLES\}$-hybrid model,  $\PiMal$ securely realizes the functionality $\FTHREEMPCABORT$ against a static, malicious adversary $\Adv$, who corrupts $P_0$.
	%----
\end{theorem}
%-----
\begin{proof}
	%-----
	Let $\Adv$ be a real-world malicious adversary corrupting  $P_0$ during $\PiMal$. We present an ideal-world adversary (simulator) $\SPiMal$ for $\Adv$,	who plays the roles of honest $P_1, P_2$  and simulates the messages received by $P_0$ during the protocol. The simulation is similar as in the semi-honest setting, where the simulator simulates $P_1, P_2$ with random inputs and keeps track of all the values that the parties (both honest and corrupt) are supposed to hold. Based on this, the simulator can find out whether the corrupt $P_0$ is sending an incorrect message(s) in any of the sub-protocols and accordingly simulates honest parties aborting the protocol. The simulator initializes	a Boolean variable $\flag = 0$, which indicates whether the honest parties abort during the simulation. Similar to the semi-honest setting, $\SPiMal$ invokes the simulator $\SSetupMal$ and learns the shared keys among $P_0$-$P_1$ and $P_0$-$P_2$, namely $k_{01}$ and $k_{02}$ and the key $k_{\Partyset}$. From the shared keys, it learns the $\pad$-masks for all the wires in $\ckt$. The details of $\SPiMal$ for the offline phase is as follows:	
	%-------
	\begin{myitemize}
		%---
		\item[--]{\em Offline Step of the instances $\PiShM$ and $\PiRecM$}: Here the simulator has to simulate nothing, as the offline phase involves no communication.
		%----
		\item[--]{\em Offline Step of the instances $\PiMultM(\Wxyzj)$}: The simulator receives $\GammaxyjB$ from $\Adv$ on behalf of $P_2$. Simulator then picks random $\MPad{{\wx}_j}, \MPad{{\wy}_j}$ and $\MPad{{\wz}_j}$ and their $\sqd$-shares on behalf of $P_1, P_2$ and honestly simulates the messages of $P_1, P_2$ as per the protocol $\PiMultM$. Namely, the simulator learns from $\Adv$ the inputs with which $P_0$ wants to call $\FTRIPLES$. If the input of $P_0$ to $\FTRIPLES$ is $\bot$, then the simulator sets $\flag = 1$, else the simulator plays the role of $\FTRIPLES$ honestly with the inputs received on behalf of $P_0$ and generates a $\shrd$-sharing of a randomly chosen multiplication triplet $(\wmd, \wme, \wmf)$. On behalf of $P_1, P_2$, the simulator sends to $\Adv$ the $\sqd$-shares of $\chi$. For the instance of $\PiTripCheck$, the simulator honestly simulates the messages of $P_1, P_2$ towards $P_0$. Moreover, the simulator sets $\flag = 1$, if it finds that $\Gammaxyj \neq \Pad{{\wx}_j} \Pad{{\wy}_j}$.        		
		%-----
	\end{myitemize}
	%-------
	
	The details of $\SPiMal$ for simulating the messages of the online phase are as follows. Informally, the simulator extracts the circuit-inputs of $P_0$ from the masked circuit-inputs which $P_0$ sends to the evaluators since the simulator will know the corresponding mask. The simulator then sets the circuit-inputs of $P_1, P_2$ to some arbitrary values and simulates the steps of the online phase. 
	%------
	During the evaluation of multiplication gates, $P_0$ receives versions of $\Mstar{\wx}$ and $\Mstar{\wy}$, which can be easily simulated as the simulator has selected them. Finally, while simulating the public reconstruction of $\shrd$-shared circuit-outputs, the simulator adjusts the shares of $P_1, P_2$, so that $P_0$ receives the same output as it would have received in the execution of the real-world protocol. 
	%------
	As done in the simulation of the offline phase, the simulator keeps track of all the values that the corrupt $P_0$ possess and sets $\flag = 1$ if it finds that $P_0$ is sending an inconsistent value during the simulated execution. 
	%------
	\begin{myitemize}
		%-----
		\item[--] {\em Online Step of the instances $\PiShM(P_i, {\wx}_j)$}:  If $P_i = P_0$, then the simulator receives $\Mask{{\wx}_j}$ and $\Mask{{\wx}_j}'$ from $\Adv$ on behalf of $P_1$ and $P_2$ respectively. The simulator sets $\flag = 1$ if it finds that $\Mask{{\wx}_j} \neq \Mask{{\wx}_j}'$, else it extracts the inputs $x_j$ of $P_0$ as $x_j = \Mask{{\wx}_j} - \Pad{{\wx}_j}$, where $\Pad{{\wx}_j}$ is the mask which the simulator learnt during the offline step. If $P_i \in \{P_1, P_2 \}$, then nothing needs to be simulated as $P_0$ does not receive any message as a part of online step of such instances of $\PiShM(P_i, {\wx}_j)$. For such instances, the simulator sets ${\wx}_j = 0$ and accordingly computes the simulated $\shr{{\wx}_j}$.
		%---
		\item[--]{\em Online Step of the instances $\PiMultM(\Wxyzj)$}: The simulator honestly performs the steps of $P_1, P_2$ for this instance and computes the simulated $\shr{{\wz}_j}$. On behalf of $P_1$, the simulator sends $\Mstar{{\wx}_j} = \Mask{{\wx}_j} + \MPad{{\wx}_j}$ and $\Mstar{{\wy}_j} = \Mask{{\wy}_j} + \MPad{{\wy}_j}$ to $\Adv$, while he sends hash of the same to $\Adv$ on behalf of $P_2$. The simulator receives $\Hash(\Mstar{{\wz}_j})$ and $\Hash(\Mstar{{\wz}_j'})$ from $\Adv$ on behalf of $P_1$ and $P_2$ respectively. The simulator sets $\flag = 1$ if $\Hash(\Mstar{{\wz}_j}) \neq \Hash(\Mstar{{\wz}_j'})$ or if $\Hash(\Mstar{{\wz}_j}) \neq \Hash(\Mask{{\wz}_j} - \Mstar{{\wx}_j} \Mstar{{\wy}_j} + \MPad{{\wz}_j})$. 		
		%-----
		\item[--]{\em Obtaining function outputs}: If $\flag$ is set to $1$ during any step of the simulation till now, then the simulator sends $\bot$ to $\FTHREEMPCABORT$, which corresponds to the case that in the real-world protocol, the honest parties abort before reaching to the output-reconstruction stage, implying that no party receives the output. Else the simulator sends inputs $\{{\wx}_j\}$ extracted on behalf of $P_0$ to $\FTHREEMPCABORT$ and receives the function outputs $\wy_1, \ldots, \wy_{\OS}$.
		%---
		\item[--] {\em Simulating the instances of $\PiRecM(\star, \Partyset )$ during the output- reconstruction}: For $j = 1, \ldots, \OS$, let $\sqr{\Pad{{\wy}_j}} = (\PadA{{\wy}_j}, \PadB{{\wy}_j})$ be the $\sqd$-shared mask, corresponding to the $j$th circuit-output, available with the simulator. Then as a part of the $j$th instance of $\PiRecM$, the simulator sends $\wy_j + \Pad{{\wy}_j}$ and $\Hash(\wy_j + \Pad{{\wy}_j})$ to $\Adv$ on behalf of $P_1$ and $P_2$ respectively. Moreover, the simulator receives $\Hash(\Pad{{\wy}_j', i})$ from $\Adv$ on behalf of $P_i$ for $i \in \EInSet$. The simulator initializes the set $I$ to $\emptyset$. If $\Hash(\PadA{{\wy}_j', i}) \neq \Hash(\PadA{{\wy}_j, i})$ then the simulator includes $P_i$ to the set $I$. The simulator then sends the set $I$ to $\FTHREEMPCABORT$ and terminates.
		%------	
	\end{myitemize}
	%---			
	The proof now follows from the fact that simulated view and real-world view of a corrupt $P_0$ are computationally indistinguishable.
	%-----
\end{proof}
%---------------------------------------------------------------------------------------------------------

%---------------------------------------------------------------------------------------------------------
We next consider the case, when the adversary corrupts one of the evaluators, say $P_1$.
%------ 
\begin{theorem}
	\label{lemma:PiMalCaseII}
	%-----------
	In the $\{\FSETUP,\FTRIPLES\}$-hybrid model,  $\PiMal$ securely realizes the functionality $\FTHREEMPCABORT$ against a static, malicious adversary $\Adv$, who corrupts $P_1$.
	%----
\end{theorem}
%-----
%-----
\begin{proof}
	%-----
	The correctness follows similar to Theorem~\ref{lemma:PiMalCaseI}. We now focus on privacy. Let $\Adv$ be a real-world malicious adversary corrupting the evaluator $P_1$ during the protocol $\PiMal$. We present an ideal-world adversary (simulator) $\SPiMal$ for $\Adv$, who plays the roles of honest $P_0, P_2$  and simulates the messages received by $P_1$ during the protocol. $\SPiMal$ invokes the simulator $\SSetupMal$ and learns the shared keys among $P_1$-$P_0$ and $P_1$-$P_2$, namely $k_{01}$ and $k_{12}$ and the key $k_{\Partyset}$. In addition, $\SPiMal$ chooses a random key $k_{02}$. The details of $\SPiMal$ for the offline phase is as follows:	
	\begin{myitemize}
		%---
		\item[--]{\em Offline Step of the instances $\PiShM$ and $\PiRecM$}: Here the simulator has to simulate nothing, as the offline phase involves no communication.
		%----
		\item[--]{\em Offline Step of the instances $\PiMultM(\Wxyzj)$}: On behalf of $P_0$, the simulator computes $\Gammaxyj = \Pad{{\wx}_j} \Pad{{\wy}_j}$. In addition, simulator learns $\GammaxyjA$ that $\Adv$ computes, for the shared key $k_{01}$. With these, simulator computes $\GammaxyjB = \Gammaxyj - \GammaxyjA$. On behalf of $P_2$, simulator computes $\MPad{{\wx}_j}, \MPad{{\wy}_j}, \MPad{{\wz}_j, 1}$ and $\MPad{{\wz}_j, 2}$ using the key $k_{12}$. The simulator receives from $\Adv$, the input with which $P_1$ wants to call $\FTRIPLES$. If this input is $\bot$, then the simulator sets $\flag = 1$. Else the simulator itself honestly performs the steps of $\FTRIPLES$ and generates $\shrd$-sharing of a random multiplication triplet $(\wmd, \wme, \wmf)$. The simulator then receives $\chi_1$ from $\Adv$ on behalf of $P_0$. The simulator then computes $\shr{\wma}, \shr{\wmb}, \shr{\wmc}$ and honestly executes the steps of $\PiTripCheck$ on behalf of $P_0, P_2$. Moreover, the simulator sets $\flag = 1$, if $\chi_1 \neq \MPad{{\wx}_j} \Pad{{\wy}_j, 1} + \MPad{{\wy}_j} \Pad{{\wx}_j, 1} + \MPad{{\wz}_j, 1} - \GammaxyjA$, else the simulator computes $\chi = \chi_1 + \chi_2$.	     	
		%----
	\end{myitemize}
	%----
	The details of $\SPiMal$ for simulating the messages of the online phase are as follows.
	%------
	\begin{myitemize}
		%-----
		\item[--] {\em Online Step of the instances $\PiShM(P_i, {\wx}_j)$}:  If $P_i = P_0$, then on behalf of $P_0$, the simulator sets $\wx_j = 0$ and sends $\Mask{{\wx}_j} = 0 + \Pad{{\wx}_j}$ to $\Adv$. Then on behalf of $P_2$, the simulator receives $\Hash(\Mask{{\wx}_j'})$ from $\Adv$, which $P_1$ wants to send to $P_2$; the simulator sets $\flag = 1$ if it finds that $\Hash(\Mask{{\wx}_j'}) \neq \Hash(\Mask{{\wx}_j})$. If $P_i = P_1$, then on behalf of $P_2$, the simulator receives $\Mask{{\wx}_j}$ from $\Adv$, which $P_1$ wants to send to $P_2$ and extract the input $\wx_j = \Mask{{\wx}_j} - \Pad{{\wx}_j}$ of $P_1$. If $P_i = P_2$, then the simulator sets $\wx_j = 0$ and sends $\Mask{{\wx}_j} = 0 + \Pad{{\wx}_j}$ to $\Adv$ on behalf of $P_2$.			
		%---
		\item[--]{\em Online Step of the instances $\PiMultM(\Wxyzj)$}: On behalf of $P_2$, the simulator honestly sends the $\sqd$-share of $\Mask{{\wz}_j}$ to $\Adv$. Then on behalf of $P_2$, the simulator receives from $\Adv$ the $\sqd$-share of $\Mask{{\wz}_j}$, which $P_1$ wants to send to $P_2$. The simulator checks if this share is correct and accordingly sets $\flag = 1$. The simulator then receives $\Mstar{{\wx}_j}$ and $\Mstar{{\wy}_j}$ from $\Adv$ on behalf of $P_0$, which $P_1$ wants to send to $P_0$. The simulator sets $\flag = 1$, if it finds that $\Mstar{{\wx}_j} \neq \Mask{{\wx}_j} + \delta_{{\wx}_j}$ or $\Mstar{{\wy}_j} \neq \Mask{{\wy}_j} + \delta_{{\wy}_j}$. On behalf of $P_0$, the simulator sends $\Mstar{{\wz}_j} =  - \Pad{{\wy}_j} \cdot \Mstar{{\wx}_j} - \Pad{{\wx}_j} \cdot \Mstar{{\wy}_j} + \delta_{{\wz}_j} + 2\gamma_{{\wx}_j {\wy}_j} + \chi$ to $\Adv$.
		%-----
		\item[--]{\em Obtaining function outputs}: If $\flag$ is set to $1$ during any step of the simulation till now, then the simulator sends $\bot$ to $\FTHREEMPCABORT$. Else the simulator sends inputs $\wx_j$ extracted on behalf of $P_1$ to $\FTHREEMPCABORT$ and receives the function outputs $\wy_1, \ldots, \wy_{\OS}$.
		%---
		\item[--] {\em Simulating the instances of $\PiRecM(\star, \Partyset )$ during the output- reconstruction}:  For $j = 1, \ldots, \OS$, let 	$(\Pad{{\wy}_j, 1}, \Mask{{\wy}_j})$ be the share of $P_1$ available with the simulator, as a part of the simulated output sharing $\shr{{\wy}_j}$. Then as a part of $\PiRecM(\shr{{\wy}_j}, \Partyset)$, on behalf of $P_2$ and $P_0$, the simulator sends $\Mask{{\wy}_j} - \PadA{{\wy}_j} - \wy_j$ and $\Hash(\Mask{{\wy}_j} - \PadA{{\wy}_j} - \wy_j)$ respectively to $\Adv$, which ensures that $\Adv$ reconstructs $\Mask{{\wy}_j} -  \PadA{{\wy}_j} - (\Mask{{\wy}_j} - \PadA{{\wy}_j} - \wy_j) = \wy_j$. On behalf of $P_0$ and $P_2$ respectively, the simulator receives $\Mask{{\wy}_j'}$ and $\Hash(\PadA{{\wy}_j}')$ from $\Adv$, which $P_1$ wants to send to $P_0$ and $P_2$ respectively as a part of $\PiRecM(\shr{{\wy}_j}, \Partyset)$. The simulator initializes the set $I$ to $\emptyset$. The simulator includes $P_0$ to $I$ if it finds that $\Mask{{\wy}_j'} \neq \Mask{{\wy}_j}$. 	Similarly, the simulator includes $P_2$ to $I$, if it finds that $\Hash(\PadA{{\wy}_j}') \neq \Hash(\PadA{y_j})$. The simulator then sends the set $I$ to $\FTHREEMPCABORT$ and terminates.			
		%------	
	\end{myitemize}
	%---						
	It is easy to see that the simulated and  real-world views of the adversary are computationally indistinguishable.
	%----
\end{proof}
%------
%----
\section{Secure Prediction}
\label{app:PrivML}

\begin{lemma}[Correctness]
	\label{lemma:PiDpM}
	%----
	In the protocol $\piBitExtM$, the following holds:
	During the offline phase, honest parties compute either $\sr = \sr_1 \sr_2$ or output $\bot$.  During the online phase, honest parties either obtain $\sign(\sr \sval)$ or output $\bot$.         
\end{lemma}
%-----
\begin{proof}
	%----   
	During the offline phase, parties locally set $\shr{\sr_1}_{P_0} = (0,0), \shr{\sr_1}_{P_1} = (\sr_1, 0)$ and $\shr{\sr_1}_{P_2} = (\sr_1,0)$, which effectively assign $\Mask{{\sr_1}} = \sr_1$ and $\Pad{{\sr_1}} = 0$. Hence, the aforementioned way of computing shares non-interactively indeed generates a valid $\shrd$-sharing of $\sr_1$ according to our sharing semantics. Similarly, the $\shrd$-sharing of $\sr_2$ is valid since the parties effectively assign $\Mask{{\sr_2}} = 0$ and $\Pad{{\sr_2}} = -\sr_2$. Given the $\shrd$-sharing of $\sr_1$ and $\sr_2$, it follows from the correctness property of protocol $\PiMultM$ (Lemma~\ref{lemma:PiMultM}) that honest parties compute either $\sr = \sr_1 \sr_2$ or output $\bot$ during the offline phase.
	
	Similar to the offline phase, following the correctness of $\PiMultM$, honest parties either compute $\shrd$-sharing of $\sr \sval$ correctly or output $\bot$ during the online phase. During the reconstruction of $\sr \sval$ towards $P_0, P_1$, since each missing share is held by two other parties and we have at most one corruption, it holds that each of $P_0, P_1$ either obtain $\sr \sval$ or output $\bot$. Now that the value $\sr \sval$ is available with both $P_0$ and $P_1$, when $P_1$ performs $\shrd$-sharing of $\msq = \MSB{\sr \sval}$, party $P_2$ can cross check hash of $\Mask{\msq}$ received from $P_1$ with the one received from $P_0$. Thus a corrupt $P_0$ or $P_1$ cannot force an honest $P_2$ to accept a wrong $\msq$ value. Moreover, the last step where parties compute $\shareB{\cdot}$-shares of $\msp \xor \msq$ is non-interactive. Hence, the correctness of online phase is ensured.
	%----
\end{proof}
\end{document}